\documentclass[a4paper]{article}
\usepackage{latexsym}
\usepackage{graphics}
\usepackage{epsfig}
\usepackage{amsmath}
\usepackage{amssymb}
\usepackage{fancybox}
\usepackage{bm}
\usepackage{amsthm}
\usepackage{mathrsfs}
\usepackage{fancyhdr}
\usepackage{ascmac}

\usepackage{mathrsfs}
\usepackage{fancyhdr}
\usepackage{ascmac}

\def\cdott{{\bm{\cdot\cdot}}}
\def\bullett{{\bullet\bullet}}

\def\iII{{I\hspace{-.1em}I}}
\def\iIII{{I\hspace{-.1em}I\hspace{-.1em}I}} 
\def\IV{{I\hspace{-.1em}V}}
\def\VI{{V\hspace{-.1em}I}}

\def\C{\mathbb{C}}
\def\Z{\mathbb{Z}}

\def\TTT{{\mathfrak T}}

\def\d{{\cal D}}
\def\DD{{\cal D}}

\def\vvv{\mathfrak{v}}
\def\VVV{\mathfrak{V}}

\def\ff{{\cal F}}

\def\f{\mathscr{F}}

\def\GG{\mathscr{G}}

\def\S{{\cal S}}
\def\SSS{{\mathfrak S}}

\def\LL{{\cal L}}

\def\LLL{{\mathfrak L}}

\def\HH{{\cal H}}

\def\A{\mathscr{A}}

\def\B{\mathscr{B}}
\def\BBB{{\mathfrak B}}

\def\R{\mathbb{R}}

\def\Varepsilon{{\mathcal E}}

\def\eee{{\mathfrak e}}

\def\I{\mathscr{I}}

\def\M{{\cal M}}
\def\MM{\mathscr{M}}

\def\N{\mathbb{N}}

\def\R{\mathbb{R}}
\def\Ri{R}
\def\RRR{\mathfrak{R}}

\def\vomega{{\mathfrak w}}
\newtheorem{definition}{Definition}[section]
\newtheorem{theorem}[definition]{Theorem}
\newtheorem{lemma}[definition]{Lemma}

\newtheorem{example}[definition]{Example}

\newtheorem{remark}[definition]{Remark}
\newcommand{\QED}{\nobreak \ifvmode \relax \else
      \ifdim\lastskip<1.5em \hskip-\lastskip
      \hskip1.5em plus0em minus0.5em \fi \nobreak
      \vrule height0.75em width0.5em depth0.25em\fi}
\begin{document}
\title{
Stochastic Metric Space and Quantum Mechanics
}
\author{Yoshimasa Kurihara\footnote{yoshimasa.kurihara@kek.jp}
\\
{\it\footnotesize The High Energy Accelerator Organization (KEK), 
Tsukuba, Ibaraki 305-0801, Japan}\\ 
}
\maketitle
\begin{abstract}
A new idea for the quantization of dynamic systems, as well as space time itself, using a stochastic metric is proposed. 
The quantum mechanics of a mass point is constructed on a space time manifold using a stochastic metric.
A  {\it stochastic metric space} is, in brief, a metric space whose metric tensor is given stochastically according to some appropriate distribution function.  
A mathematically consistent model of a space time manifold equipping a stochastic metric is proposed in this report. 
The quantum theory in the local Minkowski space can be recognized as a classical theory on the stochastic Lorentz-metric-space.
A stochastic calculus on the space time manifold is performed using white noise functional analysis.
A path-integral quantization is introduced as a stochastic integration of a function of the action integral, and  it is shown that path-integrals on the stochastic metric space are mathematically well-defined for large variety of potential functions.
The Newton--Nelson equation of motion can also be obtained from the Newtonian equation of motion on the stochastic metric space.
It is also shown that the commutation relation required under the canonical quantization is consistent with the stochastic quantization introduced in this report.
 
The quantum effects of general relativity are also analyzed through natural use of the stochastic metrics. 
Some example of quantum effects on the universe is discussed.
\end{abstract}

%%%%%%%%%%%%%%%%%%%%%%%%%%%%%%%%
% Section 1: Introduction
%%%%%%%%%%%%%%%%%%%%%%%%%%%%%%%%
\section{Introduction}\label{intro}
Although quantum mechanics (QM) is generally recognised as the most fundamental current theory describing nature, our understanding of QM remains incomplete. 
The most characteristic and mysterious aspect of QM is its requirement for a probabilistic treatment of dynamic systems. 
Unlike classical mechanics, in which the exact future configuration of a system is completely deterministic, and can in principle, be predicted if the initial conditions are sufficiently understood, physical quantities can only be probabilistically predicted in QM.
Once this probabilistic aspect is accepted via frameworks such as, the Copenhagen interpretation, the time evolution of  physical a system (its state) can be calculated exactly and its  physical observables predicted probabilistically.
Even though QM successfully explain the behavior of nature at roughly the atomic level and below, the probabilistic aspects of QM are still not understood completely.  

Several approaches have been proposed for constructing QM based on statistical/stochastic theory. 
In 1966, Nelson---who modeled the motion of a mass point as a Brownian motion under an external force---gave the initial approach in this direction.
The Newton--Nelson equation is given as a stochastic differential equation, that can be interpreted as a stochastic modification of the Newtonian equation of motion with Gaussian white noise. 
The Newton--Nelson equation has been shown to be equivalent to the Shr\"{o}dinger equation in some frameworks.
This general approach is called the {\it stochastic quantization} (SQ) method. 
Following Nelson's pioneering work, another SQ approach, the {\it stochastic space time} method, was proposed by Frederic\cite{PhysRevD.13.3183} in 1976 and has been pursued further by Ali, Prugove{$\breve{\rm c}$}ki and other authors\cite{prugovecki1984stochastic}. 
A key concept of this approach is to introduce {\it stochastic phase space}, which uses a coordinate probability density to represent the probability that a mass point is observed at two different phase points. 
Another SQ method to quantize field theories was put forth by Parisi and Wu\cite{parisi1981perturbation} in 1981, who proposed a quantum theory involving a hypothetical stochastic process with a fictitious ``{\it time}'', that differs from ordinary time.
In their theory, the conventional quantum field theory is reproduced as the thermal equilibrium limit of a stochastic system. 
Their method has been applied to many dynamical systems, for instance by Namiki, et al.\cite{NamikiMikio:1993-04-10}.
A summary of stochastic quantization can also be found in a review article\cite{1987PhR...152..227D}.

In this report, a new quantization method based on the mathematical theory of probability is proposed. 
The concept is developed as follows:
We consider the decay process of a given radioisotope. 
Because the probability of observing a decay during a unit of time is constant, the number of decays observed during a given time interval follows a Poisson distribution. 
Using this phenomenon, a clock in which the second hand advances each time a decay observed can be constructed; hereafter, this will be referred to as a {\it Poisson-clock}. 
We assume for simplicity that the Poisson-clock is designed to advance one tick per second on average. 
We then compare this clock to an ordinary mechanical clock, in which the time interval per tick of the second hand is constant. From the point of view of an observer using the mechanical clock, the second hand of the Poisson-clock seems to move randomly; however, this is of course a relative observation tied to the reference frame of the mechanical clock. 
If instead the time measured by the Poisson-clock is defined as the regular interval, the running of the mechanical clock becomes random.
A distribution of ``one second'' of the Poisson-clock, as measured by the mechanical clock, becomes an exponential distribution with an average value of unity.
Following the central limit theorem, the deviation between the Poisson and the mechanical clock after $n$ seconds will have a Gaussian distribution around zero with a variance of $n$. 
Using the mechanical clock to measure the time-of-flight of a free particle following a classical inertial path will result in a constant measured velocity.
On the other hand, if the Poisson-clock is used, measurement becomes a stochastic-process based on the Wiener measure and can be expressed using a stochastic differentiation equation. 
It has been shown that such as expression agrees with the stochastic equation obtained by Nelson\cite{PhysRev.150.1079}  that is used in stochastic quantization. 
Thus, classical mechanics with a Poisson-time measure results in QM, which suggests a new quantization method---{\it Stochastic Metric Quantization}(SMQ).
This observation can be extended to spatial coordinates as well, and an equal treatment of space and time is necessary to apply this method to relativistic quantum field theories. 
A quantum field theory can be given on the stochastic metric space, not only for flat spaces such as Minkowski space, but also for highly curved spaces such as the surface of the black hole.
As applications of this method, quantum effects in the early universe can be analyzed.

A main purpose of this work is to give a new framework of a quantum theory using mathematical tools of the stochastic metric space.
In other words, a new stochastic quantization method is proposed in this report.
A concept of our method is, in summary, that {\it classical mechanics in the stochastic space is equivalent to quantum mechanics on the standard space time manifold}.
This concept can not answer a question why quantum mechanics requires a probabilistic interpretation (the Born rule), but it can answer what is an origin of a probabilistic nature.
While our stochastic quantization gives consistent results to those from the standard method, it gives a new insight of quantum phenomenon. 
Moreover, a system which can not be quantized yet, e.g. gravitation, may be quantized using this stochastic quantization method.

This report is organized as follows:
In Section \iII,  a heuristic example showing a consequence of the stochastic metric space is given to enable better understanding of the mathematical discussion in following sections.
 In Section \iIII, the stochastic metric space is introduced and applied to the Minkowski manifold, as a model of a physical space time.
In section \IV, the stochastic quantization for a mass point is introduced and the relation among the stochastic and other quantizations is discussed.
In Section V, the effect of the stochastic metric on the space time manifold itself is discussed; this effect is found to be a quantum effect of general relativity.
In addition, stochastic effects on the Universe is presented.
The results of the stochastic quantization are summarized in Section \VI.

%%%%%%%%%%%%%%%%%%%%%%%%%%%%%%%%
% Section 2: 
%%%%%%%%%%%%%%%%%%%%%%%%%%%%%%%%
\section{Heuristic example}\label{sec2}
Before presenting a detailed discussion, a heuristic introduction of SMQ is provided in this section.
The central feature of SMQ is as follows:
when the classical  mechanics is defined on a stochastic metric space, statistical fluctuations appearing on the classical system may be  observed as quantum mechanical effects.
In this section, the hydrogen atom is treated as an example of how the stochastic metric works on a classical system.
%%%
\subsection{The Lamb shift: Welton's method}
The hydrogen atom is the first and simplest system to which QM was successfully used to provide a quantitative structural explanation.
In 1947, W. E. Lamb and R. C. Retherford reported an energy difference between the states, $^2S_{1/2}$ and $^2P_{1/2}$\cite{PhysRev.72.241} that does not appear in the lowest perturbative solution of the Dirac equation.
This energy shift is referred to the {\it Lamb shift}.
Eventually, quantum electrodynamics (QED) would be used as a successful application of renormalization calculus to calculate a precise value of the Lamb shift as a higher order correction.
Before renormalization calculus was established, however, in 1948 T$.$ A$.$ Welton gave an intuitive explanation \cite{PhysRev.74.1157} of the Lamb shift as the mean square amplitude of an electron coupled with the zero-point fluctuations of an electromagnetic field.
Welton used the following method to  derive the Lamb shift.
The potential energy, $V(r)$, is given as an isotropic function of the distance $r=|{\bm r}|$, where ${\bm r}$ is the three-dimensional vector of the electron's position relative to the proton. 
The potential energy can be Taylor-expanded with respect to the small displacement of the electron position $\bm{r}\rightarrow\bm{r}+\delta\bm{r}$ as;
\begin{eqnarray*}
V(\bm{r}+\delta\bm{r})&=&\left(1+\delta\bm{r}\cdot\nabla
+\frac{1}{2}\left(\delta\bm{r}\cdot\nabla\right)^2 
+\cdots\right)V(\bm{r}).
\end{eqnarray*}
When the displacement is identified as a variance of the electron position arising from the quantum fluctuation, an averaged value with respect to the vacuum state can be estimated as
\begin{eqnarray*}
\left\langle\delta\bm{r}\cdot\nabla\right\rangle&=&0,\\
\left\langle\left(\delta\bm{r}\cdot\nabla\right)^2\right\rangle&=&\frac{1}{3}\left\langle\delta\bm{r}^2\right\rangle\nabla^2,
\end{eqnarray*}
where the factor of $1/3$ in the second relation comes from averaging over the three spatial dimensions.
The hydrogen atom has a Coulomb potential energy $V(r)=-\alpha/r$, where $\alpha=1/137.0359895$ is the fine structure constant at the low energy limit.
The average of the potential energy fluctuation can be written as
\begin{eqnarray*}
\left\langle\delta V\right\rangle&=&\left\langle V(\bm{r}+\delta\bm{r})-V(\bm{r})\right\rangle
=
\frac{1}{6}\left\langle\left(\delta\bm{r}\right)^2\right\rangle
\left\langle\Delta\left(-\frac{\alpha}{r}\right)\right\rangle_{H},
\end{eqnarray*}
where $\langle\bullet\rangle_{H}$ is an average over the hydrogen atom.
The average of a square fluctuation of the electron's position can be estimated as
\begin{eqnarray*}
\left\langle\left(\delta\bm{r}\right)^2\right\rangle&=&\frac{2\alpha^3a_B^2}{\pi}\int^{k_1}_{k_0}\frac{dk}{k},
\end{eqnarray*}
where $a_{B}=1/(m_e\alpha)$ is the Bohr radius of the hydrogen atom and $m_e$ is the electron mass.
The upper and lower bounds of the integration must be fixed according to quantum field theoretic considerations.
The average of the Laplacian on the potential energy over the hydrogen atom can be estimated as
\begin{eqnarray*}
\left\langle\Delta\left(-\frac{\alpha}{r}\right)\right\rangle_H&=&4\pi\alpha\Big|\psi(0)\Big|^2,
\end{eqnarray*}
where $\psi(r)$ is the wave function of the hydrogen atom.
Here well-know relation $\Delta(1/r)=-4\pi\delta(\bm{r})$ is used; consequently.
The average value of the quantum fluctuation of the potential energy can be obtained as;
\begin{eqnarray*}
\left\langle\delta V\right\rangle&=&\frac{4}{3}\alpha^4a_B^2|\psi(0)|^2\log{\frac{k_1}{k_0}}.
\end{eqnarray*}
The $S$-wave solution for the hydrogen atom at the origin is then given as;
\begin{eqnarray*}
\psi_{n}(0)=\left[\pi\left(na_B\right)^3\right]^{-1/2},
\end{eqnarray*}
 where $n$ is the main quantum number.
 From field theoretic considerations, the lower and upper cutoffs of the electron momentum are chosen to be $k_0=\alpha^2m_e$ and $k_1=m_e$ respectively \cite{landau1977course}. 
 To enable appropriate setting of these cutoffs, UV and IR divergences are avoided in this method.
Finally, the deviation of the potential energy for the $S$-wave solution with the main quantum number $n$ is obtained as;
\begin{eqnarray*}
\left\langle\delta V\right\rangle&=&\frac{4\alpha^4}{3\pi n^3a_B}\log{\frac{1}{\alpha^2}}.
\end{eqnarray*}
This result is consistent with a more precise result\cite{landau1977course} obtained using perturbative QED, which is;
\begin{eqnarray*}
\left\langle\delta E_{QED}\right\rangle&=&\frac{4\alpha^4}{3\pi n^3a_B}\left(\log{\frac{1}{\alpha^2}}
+L_n+\frac{19}{30}
\right).
\end{eqnarray*}
Where the values of the correction term, $L_n$, are given in Table \ref{t1}. 
 
%%%%%%%%%%%%%%%%%%%
\begin{table}[b]
	\begin{center}
		\begin{tabular}{c|ccccc}
 			$n$ & $1$ & $2$ & $3$ & $4$ &  $\infty$\\
 			\hline
 			$L_n$ & $-2.984$ & $-2.812$ & $-2.768$ & $-2.750$ & $-2.721$\\
		\end{tabular}
		\caption{Numerical values of the correction term, $L_n$\cite{landau1977course}.}\label{t1}
	\end{center}
\end{table}
%%%%%%%%%%%%%%%%%%%
%%%
\subsection{Stochastic metric point of view}
Welton's method suggests that the quantum effects on hydrogen energy are caused by the variance of the electron position.
This scenario can be realized more naturally using the stochastic metric space.
As is well known in quantum field theory, the lowest result in the perturbative calculations is the same as the result of classical mechanics, and then a quantum effect will appear starting from one loop corrections.
Let us apply a stochastic method on the lowest Bohr solutions of hydrogen energy. 

From the stochastic metrical point of view, the distance between two points (the proton and electron positions) is given as a stochastic variable.
When distances are given stochastically, their distribution must be Gaussian following the central limit theorem, variance of the distance possibly be  proportional to distance. 
Here, we assume a variance $\sigma_n^2$, given as;
\begin{eqnarray}
\sigma_n^2&=&\frac{a_B^2}{3n}\left(\frac{\alpha}{\pi}\right)^2,\label{sigman}
\end{eqnarray}
where the factor $1/3$ reflects the dimensionality of space, the standard deviation of the Gaussian distribution is taken as proportional to $a_B/\sqrt{n}$, and a verification of the factor $(\alpha/\pi)^2$ will be given later in this section.
Hereafter, the atomic unit $a_B=1$ will be used in this section.
The statistical fluctuation of the potential energy can be written as;
\begin{eqnarray*}
\delta V&=&V(r(1+\delta_r))-V(r),\\
&=&-\frac{\alpha}{r}\left(1+\delta_r\right)^{-1}+\frac{\alpha}{r},\\
&=&-\frac{\alpha}{r}\left(
-\delta_r+\delta_r^2-\delta_r^3+\cdots
\right).
\end{eqnarray*}
The relative  deviation is then obtained as
\begin{eqnarray*}
\frac{\delta V}{V}&=&-\delta_r+\delta_r^2+\mathcal{O}(\delta_r^3).
\end{eqnarray*}
When the distance $r$ is given stochastically, the fluctuation $\delta_r$ should have a Gaussian distribution.
Thus, the average value of the relative deviation of the potential energy is given as;
\begin{eqnarray}
\left\langle\frac{\delta V}{V}\right\rangle&=&\frac{1}{\sqrt{2\pi}\sigma_n}\int_{-\infty}^{+\infty} d\delta_r
\frac{\delta V}{V}
\exp{\left(
-\frac{\delta_r^2}{2\sigma_n^2}
\right)},\nonumber\\
&=&\sigma_n^2+\mathcal{O}(\sigma_n^4),\label{dvov}
\end{eqnarray}
where $\sigma_n$ is given by (\ref{sigman}).
In reality, the Gaussian mean of the reciprocal variable
\begin{eqnarray}
\left\langle\frac{1}{r}\right\rangle&=&\frac{1}{\sqrt{2\pi}\sigma_n}\int_{0}^{+\infty} dr
\frac{1}{r}
\exp{\left(
-\frac{\left(r-\bar{r}_n\right)^2}{2\sigma_n^2}
\right)},\label{gmailri}
\end{eqnarray}
does not exist because the integration on the right-hand side does not converge.
Here, $\bar{r}_n$ is the mean radius of $S$-wave hydrogen with the main  quantum number $n$, which is considered to correspond to the UV divergence of the QED because this divergence comes from the short distance limit $r\rightarrow0$.
In our calculation, the divergence is avoided by looking at variances only around the average value.
If the integration in (\ref{gmailri}) is performed numerically only around the mean value (from which the main contribution must come), results consistent with (\ref{dvov}) can be obtained. 

We can now compare the numerical results of the stochastic metric (\ref{dvov}) with those of QED.
Because the $S$-wave hydrogen energy at the Born level is $\langle E_{QED}\rangle=-\alpha a_B/2n^2$, the relative correction on the hydrogen energy  can be given as;
\begin{eqnarray*}
\delta_{QED}&=&
\left|\frac{\left\langle\delta E_{QED}\right\rangle}{\langle E_{QED}\rangle}\right|
=
\frac{8\alpha^3}{3\pi n}\left(\log{\frac{1}{\alpha^2}}
+L_n+\frac{19}{30}
\right).\label{LambSLL}
\end{eqnarray*}
The numerical results are summarized in Table \ref{t2}, along with those from SMQ.
%%%%%%%%%%%%%%%%%%%
\begin{table}[b]
	\begin{center}
		\begin{tabular}{c|ccccc}
 				$n$ & $2$ & $3$ & $4$ & $100$	\\
			\hline
				$\delta_{QED}$ & $1.264\times10^{-6}$ & $8.473\times10^{-7}$ & $6.369\times10^{-7}$ & $2.557\times10^{-8}$\\
  			
				$\delta_{SMQ}$ & $1.272\times10^{-6}$ & $8.478\times10^{-7}$ & $6.359\times10^{-7}$ & $2.543\times10^{-8}$\\
			%\hline
				$\delta_{SMQ/QED}$  & $+0.64\%$ & $+0.07\%$ & $-0.17\%$ & $-0.53\%$\\   
		\end{tabular}
		\caption{
Energy corrections of $S$-wave hydrogen.
Numerical values of $\delta_{QED}$ and $\delta_{SM}$ are obtained using QED and SMQ, respectively.
As seen in the last row of the table, the agreement between the two methods, $\delta_{SMQ/QED}=(\delta_{SMQ}-\delta_{QED})/\delta_{QED}$, is excellent.
		}\label{t2}
	\end{center}
\end{table}
%%%%%%%%%%%%%%%%%%%
A relative correction from SMQ, $\delta_{SMQ}$, is defined as 
\begin{eqnarray*}
\delta_{SMQ}&=&\sqrt{2}\left\langle\frac{\delta V}{V}\right\rangle,
\end{eqnarray*}
in which the factor $\sqrt{2}$ appears because two independent statistical fluctuations of $\delta V$'s with $n=1$ and $n\geq2$ are, in principle, involved with $\delta_{SMQ}$.
For the calculation of $\delta_{QED}$ at $n=100$, the value of $L_n$ at $n=\infty$ is used.
The result at $n=100$ is shown to examine the asymptotic behavior of the correction at a large distance.
As seen in Table \ref{t2}, the agreement is excellent, which justifies the assumed factor $(\alpha/\pi)^2$ in the variance (\ref{sigman}).

%%%%%%%%%%%%%%%%%%%%%%%%%%%%%%%%
% Section 3: Probabilistic Minkowski manifold
%%%%%%%%%%%%%%%%%%%%%%%%%%%%%%%%
\section{Introduction of stochastic metric space}\label{sPM}
The metric cannot be completely random; if the distance between two points on the space time manifold is completely random without any restrictions, the associated theory might not respect causality, in which it would not be useful for describing nature.
The SMQ method must therefore be constructed on a legitimate mathematical base.

Study of the mathematical theory of the stochastic metric was initiated by Menger\cite{menger1942statistical} in 1942, whose work was immediately followed by that of Wald\cite{wald1943statistical} in 1943. 
After intensive efforts by Schweizer and Sklar\cite{schweizer1960} in the late 1950s, $\breve{\rm S}$erstnev introduced  the (distance) distribution and triangle functions\cite{Serstnev1962} in 1963.  
To honor his pioneering work, the stochastic (statistical) metric space is generally called the {\it Menger space}; in this report, we refer to it as the {\it Stochastic metric-space}(SM-space).
Mathematical research on the SM-space is ongoing, albeit with limited application to studies in the field of physics.
One example of physical research is  the relation between stochastic phase spaces and SM-spaces, as discussed by Guz\cite{Guz1984-GUZSPS}, the mathematical preparation of which was given at the beginning of this section.
%
%Stochastic metric space
%
\subsection{Stochastic metric space}\label{sPMS}
Here basic definitions of the SM-space are introduced based primarily on $\breve{\rm S}$erstnev\cite{Serstnev1962}, Schweitzer and Sklar\cite{schweizer1983probabilistic}.
The distribution and triangle functions play essential roles in characterizing the SM-space.
A {\it distribution function} $F$, is a non-decreasing function defined as  a map such that;
\begin{eqnarray*}
F:\R\rightarrow[0,1]:s\mapsto F(s),
\end{eqnarray*}
with 
\begin{eqnarray*}
F(-\infty)=0,~F(\infty)=1,&~&
~{\rm and}~
s<t\Rightarrow F(s)\leq F(t).
\end{eqnarray*}
The distribution function satisfying $F(0)=0$, which is referred to as the {\it distance distribution function}, plays the role of a cumulative (probability) distribution function to obtain a distance, $s$.
A set of distance distribution functions is denoted as $\Delta^+$, and the triangle function, $T$, is introduced to give a natural extension of the triangle inequality on Euclidean space.
$T$ is defined as a map such as;
\begin{eqnarray*}
T:\Delta^+\otimes\Delta^+\rightarrow\Delta^+:\{F,G\}\mapsto T(F,G),
\end{eqnarray*}
which satisfies the following for any distribution functions $F,G,H,K\in\Delta^+$,
\begin{enumerate}
\item $T(F,G)=T(G,F)$;
\item $T(F,G)\leq T(H,K)$, whenever $F\leq H$,  and $G\leq K $;
\item $T(F,\Theta)=F$, where $\Theta(s)$ is the Heaviside unit function.
Note that $\Theta(s)\in\Delta^+$;
\item $T(T(F,G),H)=T(F,T(G,H))$.
\end{enumerate}
Here, $F\leq H$ means $F(s)\leq H(s)$ for any $s>0$.
The SM-space can be defined as follows using the above two functions.
\begin{definition}{\bf (Stochastic metric space)}\label{PMS}\\
A stochastic metric space  is a triple $(\S,\f,T)$ of
\begin{itemize}
\item $\S$: a set of points on the space;
\item $\f$: a map $\f:\S\otimes\S\rightarrow\Delta^+$, {\rm and};
\item $T$: a triangle function,
\end{itemize}
which satisfies for any points $p,q,r\in \S$ and $s,s'\in\R$;
\begin{enumerate}
\item $\f(p,p)(s)=\Theta(s)$;
\item $\f(p,q)(s)\neq\Theta(s)$, if $p\neq q$;
\item $\f(p,q)(s)=\f(q,p)(s)$, {\rm and};
\item $\f(p,r)(s+s')\geq T(\f(p,q)(s),\f(q,r)(s'))$.
\end{enumerate}
\end{definition}\noindent
If Requirement~{\it 2} is omitted, the space is referred to as a stochastic {\it pseudometric} space. 
Requirement {\it 4} is a type of triangle inequality originally introduced by Menger\cite{menger1942statistical}. 

%
%Stochastic Minkowski manifold
%
\subsection{Stochastic Minkowski manifold}\label{sPMM}
To utilize the SM-space for relativistic field theories, the stochastic metric should be introduced on a Minkowski manifold. 
The Minkowski manifold equipping the stochastic metric is referred to as the {\it stochastic Lorentz metric space} (SLM-space) hereafter. 
A distance between two points on the SLM-space will fluctuate around the geometric distance measured by the Lorentz metric, with the distance measured by the (non-fluctuating) Lorentz metric referred here as the {\it geometrical distance}.
A distribution function is required to give the geometrical distance as an average value over a two-point ensemble on the SLM-space, with the variance of the distribution function set to be proportional to its geometrical distance.
This distribution function makes a null vector (vector with zero length) without any fluctuation, a desirable characteristic for restraining a null photon mass after quantum corrections.
Furthermore, to satisfy the causality condition the probability changing a sign of the length of a string stretching between two points must be zero.
The SLM-space satisfying above requirements can be introduced as follows:
\begin{definition}{\bf (Stochastic Lorentz metric space)}\label{PMM}\\
A stochastic Lorentz-metric-space (SLM-space) is a triple of $(\M,F_{e},T)$ such that:
\begin{enumerate}
\item $\M$ is a $4$-dimensional Minkowski manifold with a metric tensor
$
\eta_{\mu\nu}=\mathrm{diag}(1,-1,-1,-1).
$
The geometrical distance between two points on $\M$, say $x^\mu$ and $y^\mu$ in the local coordinate system is defined as
\[
\d(x-y)=\sqrt{|\eta_{\mu\nu}(x-y)^\mu(x-y)^\nu|},
\]
using the Lorentz-metric tensor.
In this report, the Einstein convention to take the sum over repeated indices is used. 
The indices run from zero to three. 
\item $F_e(s;d)$ is  a map of 
$
\{d;s\}\mapsto F_e(s;d)\in[0,1], 
$
where $F_e$ is chosen as the following exponential probability density;
\begin{eqnarray*}
F_e(s;d)&=&
\begin{cases}
0&s<0,\\
1-\exp{\left(-\frac{s}{d}\right)}&s\geq0.
\end{cases}
\end{eqnarray*}
\item $T$ is  the triangle function defined as  
\begin{eqnarray*}
T\left(F_e\left(s;\d(x-y)\right),F_e\left(s';\d(y-z)\right)\right)
=F_e\left(s;\d(x-y)\right)F_e\left(s';\d(y-z)\right),
\end{eqnarray*}
for any $x,y,z\in\M$ and $s,s'>0$.
\end{enumerate}
\end{definition}\noindent
The function $F_e\left(s; \d(x-y)\right)$ introduced above can be interpreted as the probability of observing a distance less than $s$ when the geometrical distance between $x$ and $y$ is $\d(x-y)$. 
In other words, the probability of finding a distance greater than $s\geq 0$ is $\bar{F}_e(s,d)=1-F_e(s,d)=\exp{\left( -s/d\right)}$ when the geometrical distance is $d$.
The SM-space, whose distribution function $\f$ is a function of $s/\d(x-y)$ only, is called the  {\it simple space}\cite{schweizer1960statistical}; the SLM-space is an example of a simple space. 
 A probability measure can be derived from $F_e\left(s; \d(x-y)\right)$ as follows. 
 The exponential probability density of a variable $s$ with a mean value of $\lambda$ is 
\begin{eqnarray*}
\mu_e(s;\lambda)&=&\frac{1}{\lambda}
\exp{\left(-\frac{s}{\lambda}\right)}ds.
\end{eqnarray*} 
Therefore, the probability measure corresponding to $F_e(s; \d(x-y))$ in the SLM-space can be written as 
\begin{eqnarray*}
dF_e\left(s;\d\left(x-y\right)\right)&=&\frac{1}{\d(x-y)}\exp{\left(-\frac{s}{\d(x-y)}\right)}ds
=\mu_e\left(s;\d(x-y)\right).
\end{eqnarray*}
An interpretation of the probability measure $\mu_e\left(s;\d(x-y)\right)$ is that it is the probability of obtaining the distance within  $[s,s+ds]$ for $\d(x-y)$. 
It can be shown that the SLM-space satisfies the above definition of a Stochastic Lorentz metric space with an additional requirement, i.e., it is a stochastic pseudometric space. Thus, we give the following remark.
\begin{remark}{\bf (the SLM-space is the SM-space)}\label{PMMS}\\
The stochastic Lorentz metric-space $(\M,F_e,T)$ is a stochastic pseudometric space. 
\end{remark}\noindent
A  proof of the above remark is given in Appendix \ref{ap1}.
The SM-space has a rich structure\cite{schweizer1960statistical} as follows: 
\begin{itemize}
\item The SM-space is a topological space with a $\lambda$-neighborhood at $p\in\M$ as follows;
\begin{eqnarray*}
N_p(\varepsilon,\lambda)
=\left\{q|F_e\left(\varepsilon;\d(p-q)\right)>1-\lambda\right\},
\end{eqnarray*}
where $\varepsilon,\lambda>0$. 
\item The SM-space is a Hausdorff space with a triangle function under the above topology. 
\item On the SM-space with the continuous triangle function, the convergence and continuity of the distribution function $F_e(s;\d(p-q))$ is ensured.
\end{itemize}
These properties ensure that the physical theory can be constructed on the base of the SLM-space following {\bf Remark \ref{PMMS}}.

In addition to the above properties,  it is essential that the physical theory must retain causality.
In relativistic theories, causality expressed by the following statement:
when the (geometrical) distance between two points $p,q\in\M$ is space like such that $|p-q|<0$, the two-point correlation function must be zero.
Our choice of distribution function, $F_e(s;\d(p-q))$, does not destroy causality because it does not change the sign of the distance; moreover, the existence of the triangle function ensures that there are no shortcuts between two space time points.

%
%The Poisson process and Gaussian measure
%
\subsection{The Poisson process and Gaussian measure}\label{poissongauss}
Consider a series of independent random variables $\{X_i\}_{i=1,2, \cdots}$ with a probability measure with $P(0<s<X_i<s+ds)=p(s)ds$.
A series of random variables $\{A_j\}_{j=0,1,2,\cdots}$ such as; $A_0=0$ and $A_j=\sum_{i=1}^jX_i$ will (almost surely) be a the {\it stochastic-process induced by the probability measure $p(s)ds$}.
It is known that the exponential measure induces the Poisson process; a detailed explanation of the Poisson process is given in Appendix \ref{ap2}.
The following remark will play an essential role in constructing the SMQ method.
\begin{remark}{\bf (Brownian motion induced by exponential measure)}\label{brownEM}\\
Suppose that the Poisson process is induced by the exponential measure $\mu_e(s,\lambda)$ with a series of random variables, $\{X_s\}_{s=1,2,\cdots}$, and 
\begin{eqnarray*}
Z_n&=&\frac{1}{\sqrt{n}}\sum_{s=1}^n(X_s-\lambda).
\end{eqnarray*}
The random variable $Z_n$ is know to show a convergence in law to a normal distribution with a mean value at zero and a variance $\lambda$. 
In particular, for $-\infty<a<b<\infty$,
\begin{eqnarray*}
\displaystyle \lim_{n \to \infty}P(a\leq Z_n\leq b)&=&\frac{1}{\sqrt{2\pi\lambda}}\int_a^b \exp{\left(-\frac{x^2}{2\lambda}\right)}dx,
\end{eqnarray*}
can be obtained.  
Brownian motion is induced by this random variable.
\end{remark}
\noindent
The proof of this remark is given in Appendix \ref{ap3}.
To summarize the above derivations, if the position of a force-free mass point at time $t$ can be represented by a series of random variable $\{X_s\}$, the mass point can follow a Brownian motion and a set of positions, each at an equal geometrical distance, should be observed as the Poisson process in the SLM-space. 

%%%%%%%%%%%%%%%%%%%%%%%%%%%%%%%%
% Section 4: 
%%%%%%%%%%%%%%%%%%%%%%%%%%%%%%%%
\section{Dynamics of a mass point in the SLM-space}\label{sec3}
Particle dynamics are introduced on the SLM-space in this section. 
Initially, the {\it path} of a mass point is defined geometrically on the classical space time manifold. 
The term {\it classical} is used in this section to denote objects that do not have a stochastic property.
Based on this initial treatment, the Lagrangian density and action integral are defined in terms of the path; in other words, paths are treated as dynamical objects. 
Quantum effects of particle dynamics, such as the quantum equation of motion and the uncertainty relation, are obtained because of the SLM-space. 
%%%
\subsection{Geometrical setup}\label{sec3-1}
One of the more important calculations in the dynamical theory of a mass point is determining its trajectory. 
The trajectory may be a continuous smooth line between two space time points as would occur under a classical theory.
Such  trajectory, called a {\it curvilinear path}, is defined geometrically on the Minkowski manifold. 
In this subsection, a classical Minkowski manifold is assumed and the stochastic aspects of the manifold are suspended for the time being.

The geometrical setup introduced in this subsection is based on reference\cite{kurihara2014}.
A set of bounded functions $\gamma$ such that
\begin{eqnarray}
&\gamma&: \tau \in [\tau_1,\tau_2] \mapsto 
{\bf \gamma}(\tau)=\left(\tau,\gamma^1(\tau),\gamma^2(\tau),\gamma^3(\tau)
\right),\label{curveg}\\
&\gamma^i&:\tau\mapsto\gamma^i(\tau),~~\{\gamma^i\}_{i=1,2,3}\in C^2(\R),\label{cvlin}
\end{eqnarray}
is called a {\it curvilinear path} (or simply {\it path}), and a set of paths, denoted by  $\Gamma$, is called a curvilinear-path space.
%\end{definition}\noindent
A parameter $\tau$ is an order parameter used to measure the length of a path from the starting point $\tau_1$, which is set to zero without any loss of generality. 
We set $\tau_1=0$ and $\tau_2=T$, and  hereafter only paths whose end points are fixed at $\gamma(0)=\xi^\mu_1=(0,\xi_1^1,\xi_1^2,\xi_1^3)=(0, \bm\xi_1)$ and $\gamma(T)=\xi^\mu_2=(T,\xi_2^1,\xi_2^2,\xi_2^3)=(T, \bm\xi_2)$ are considered.
The dynamics of a mass point cannot be determined using only information on the path itself; the time-derivative along the path is also necessary. 
A velocity vector is a tangent vector at a point $\gamma(\tau)$ along the curvilinear path, defined as
\begin{eqnarray}
\frac{d\gamma}{d\tau}(\tau)&=&
\left(1,\frac{d\gamma^1}{d\tau}(\tau),\frac{d\gamma^2}{d\tau}(\tau),
\frac{d\gamma^3}{d\tau}(\tau)\right).
\end{eqnarray}
Hereafter, the velocity vector is written using  shorthand notation as $\dot{\bf \gamma}(\tau)=(1,\dot{\gamma}^1(\tau),\dot{\gamma}^2(\tau),\dot{\gamma}^3(\tau))$. 
The velocity vector can be expressed in terms of natural bases on a tangent space $T_\tau\gamma$ at $\gamma(\tau)$ as
\begin{eqnarray}
\dot{\bf \gamma}^i(\tau)=
\frac{d\gamma^i(\theta)}{d\theta}
\frac{\partial}{\partial \gamma^i}\Big|_{\theta=\tau},\label{velo}
\end{eqnarray}
where $i$ runs from $1$ to $3$ and is not summed in the right-hand side of (\ref{velo}).
A tangent vector bundle $\dot{\Gamma}=\bigcup_{\gamma(\tau)\subset\gamma\in\Gamma}T_t\gamma$ is referred to a velocity bundle.  

All the information needed to determine the dynamics of the system is contained in the Lagrangian, which can be defined as a map from $\{\Gamma, \dot{\Gamma}\}$ to a smooth function such that
\begin{eqnarray}
\LL:\Gamma\otimes\dot{\Gamma}
\rightarrow C^{\infty}:\{\gamma,\dot{\gamma}\}
\mapsto \LL(\gamma,\dot{\gamma}).
\end{eqnarray}
A new dynamical object, the canonical momentum, can be derived from the Lagrangian as;
\begin{eqnarray}
\pi^i=\frac{\delta\LL(\gamma,\dot{\gamma})}{\delta\dot{\gamma}_i},
\end{eqnarray}
where $\delta$ denotes the functional derivative.
A tuple of $\GG=\{\gamma,\pi\}$ is referred to as a {\it phase space}.

The Hamiltonian can be derived from the Lagrangian and the canonical momentum as;
\begin{eqnarray*}
\HH:\GG\rightarrow C^{\infty}:
\{ \gamma,\pi \} &\mapsto&
\HH\left(\gamma(t),\pi(t)\right),
\end{eqnarray*}
where
\begin{eqnarray}
\HH\left(\gamma(t),\pi(t)\right)
&=&\left(\pi_i \dot{\gamma}^i-\LL(\gamma,\dot{\gamma})\right)\Bigl|_{\dot{\gamma}=\phi}
=\pi_i \phi^i(\gamma,\pi)-\LL(\gamma,\phi(\gamma,\pi)).
\label{LagHam1}
\end{eqnarray}
The function $\phi^i(\gamma,\pi)$ gives the solutions of 
\begin{eqnarray}
\pi_i-\frac{\delta \LL}{\delta \dot{\gamma}^i}=0,~~~~i=1,2,3\label{LagHam2},
\end{eqnarray}
with respect to the velocity vector of $\dot{\gamma}^i=\phi^i(\gamma,\pi)$.
Here, the Lagrangian is assumed to be a holomorphic function and therefore a solution $\phi^i(\gamma,\pi)$ always exists and the Hamiltonian is defined on the phase space.
The existence of  the inverse Legendre transform is ensured  because
\begin{eqnarray}
\LL(\gamma,\dot{\gamma})
&=&\tilde{\phi}_i(\gamma,\dot{\gamma})\dot{\gamma}^i-\HH\left(\gamma,\tilde{\phi}(\gamma)\right)\label{LagHam3},
\end{eqnarray}
where $\tilde{\phi}$ is a solution of the equation of
$\dot{\gamma}^i-\delta\HH/\delta\pi_i=0$, that is
\begin{eqnarray}
\dot{\gamma}^i&=&\frac{\delta\HH}{\delta\pi_i}\Big|_{\pi=\tilde{\phi}}.\label{gammadot}
\end{eqnarray}
The action integral can be introduced using the Lagrangian or Hamiltonian as follows: 
a map $\I(\gamma)$ from a path $\gamma$ to a real number $\R$ such as
\begin{eqnarray*}
\I:\Gamma\otimes\dot{\Gamma}\rightarrow\R:
\{\gamma,\dot{\gamma}\} &\mapsto&\I(\gamma,\dot{\gamma}),
\end{eqnarray*}
where
\begin{eqnarray}
\I(\gamma,\dot{\gamma})&=&
\int_0^Td\tau~\LL\left(\gamma(\tau),\dot{\gamma}(\tau)\right)
=\int_0^Td\tau~\left(
\pi_i\dot{\gamma}^i-\HH\left(\gamma,\pi\right)\right)\Big|_{\pi=\tilde{\phi}}\label{action}
\end{eqnarray}
is referred to an {\it action integral} or simply an {\it action}.
As is well known, the curvilinear path that gives $\delta\I=0$ satisfies the canonical equations of motion;
\begin{eqnarray}
\frac{d\gamma_c^i}{d\tau}&=&\frac{\partial\HH}{\partial\pi_{ci}},~~~~
\frac{d\pi^i_c}{d\tau}=-\frac{\partial\HH}{\partial\gamma_{ci}},\nonumber\\
\frac{d\HH}{d\tau}&=&\frac{\partial\HH}{\partial\tau}.
\end{eqnarray}
A solution of the equation of motion is called the {\it classical path} and denoted by $\gamma_c(\tau)$. 
The classical velocity and momentum can be further defined from the solution of the equation of motion and denoted as $\dot{\gamma}_c$ and $\pi_c$, respectively. 
%%%
\subsection{Stochastic quantization of classical paths}\label{sec3-2}
In the Lagrangian formalism of classical mechanics, the principle of least action (Hamilton's principle) is  the most fundamental principle.
The behavior of a dynamical system can be described by the Euler--Lagrange, or canonical equations of motion; or the Hamilton--Jacobi equation, which can be derived from the Lagrangian, Hamiltonian or action integral, respectively. 
In the traditional method, QM is introduced only after establishing the classical dynamics of the system by requiring some quantization conditions on the system.
However, if QM is the most fundamental theory of nature, a dynamical system  should be formulated using QM as a precedential theory to classical mechanics.
Following this, some appropriate approximation can be used to derive classical dynamics from the quantum system.
Here, we propose a new algorithm establishing a quantum system preceding a classical system.
(There is another approach to extract quantum mechanics directly using the principle of least action\cite{Yasue3,Yasue1981327}.) 
As our approach is based on the stochastic property of the metric, the principle of least action is not the most fundamental principle but is instead a theorem that can be derived  from a more fundamental stochastic principle.  
To replace the principle of least action, we propose using the {\it entropy-extremal principle} to extract dynamics from the system that equips the Lagrangian.

In the following, we first show that path-integral quantization can be deduced from the entropy-extremal principle.
Then, yet another method---Newton--Nelson equation\cite{gliklikh}---to extract the quantum equation of motion is given for the SLM-space.

\subsubsection{Path integral on the SLM-space}
A mass point traveling in the SLM-space will follow a trajectory that, owing to the probabilities of the SLM-space, would be expected to be a stochastic process. 
Such a trajectory cannot be a smooth curvilinear path as defined in the previous subsection; correspondingly, the instantaneous velocity of the mass point will not exist in classical sense and must be redefined in terms of probability theory. 

Let us consider a following integration:
\begin{eqnarray}
\gamma^i_c(\tau)-\gamma^i_c(0)&=&
\int_{\gamma(0)}^{\gamma(\tau)}d\gamma^i_c\label{gc}
=\int_0^\tau\dot{\gamma}^i_c(t)dt,\label{gc2}
\end{eqnarray}
where $i=1,2,3$ is the spatial components of the vectors.
The parameter $\tau$ is an ordering parameter along the curvilinear path, functioning as, e.g. tick marks along the path. 
We then consider the integration of a bounded function, $f(x)\in C^2$ along the classical path; 
%Stieltjes
\begin{eqnarray}
I&=&\int_{\gamma(0)}^{\gamma(\tau)}f(\gamma_c)\Bigl|d\gamma_c\Bigr|
=\int_0^\tau f\left(\gamma_c(t)\right)\Bigl|{\dot \gamma_c}(t)\Bigr|dt,\label{fint2}
\end{eqnarray}
where the integration measure is given as
\begin{eqnarray}
\Bigl|d\gamma_c\Bigr|&=&\sqrt{|\eta_{\mu\nu}d\gamma_c^\mu d\gamma_c^\nu|~}
=\sqrt{|\eta_{\mu\nu}{\dot \gamma}_c^\mu(t){\dot \gamma}_c^\nu(t)|~}~dt,\nonumber\\
&=&\Bigl|{\dot \gamma_c}(t)\Bigr|dt.\label{dgammac}
\end{eqnarray}
We note that these lengths are also bounded under classical mechanics, and, therefore, that integrations (\ref{gc}) and  (\ref{fint2}) are well-defined as Riemann-Stieltjes integrations.
The classical path and  velocity, $\gamma_c(\tau)$ and $\dot{\gamma}_c(\tau)$, respectively, are functions of $\tau$ under the Lagrangian formalism.
On the SLM-space, the trajectory of a mass point is cast as a stochastic process induced by the exponential measure; therefore, the motion of a mass point in the SLM-space can be treated as Brownian motion following {\bf Remark \ref{brownEM}}. 
We consider, e.g. a mass point stopping at the origin of the coordinate system in classical mechanics.
The mass point remains at the same spatial position  for a length of time $\tau$.
The geometrical distance between the initial and final points is also $\tau$. 
In the SLM-space, distance is a random variable and has a Gaussian distribution with a standard deviation proportional to $\sqrt{\tau}$ following the stipulation of {\bf Remark \ref{brownEM}}  that there is no motion other than Brownian motion.

The quantum effects on the mass point appear because of the stochastic treatment for the path. 
The classical path and  velocity are replaced by random variables induced by the exponential measure of the SLM-space. 
The It{\^ o} process, defined in {\bf Definition \ref{ip}}, is introduced as
\begin{eqnarray}
\gamma^i_\omega(\tau)&=&\gamma^i_c(0)+
\int_0^\tau\dot{\gamma}^i_\omega(t)dt
=\gamma^i_c(0)+\int_0^\tau\sigma^i_k(t)\cdot dB^k(t),\label{sdg}
\end{eqnarray}
in analogy with equation (\ref{gc2}). 
A precise definition of (\ref{sdg}) can be found in Appendix \ref{ap4}.
We note that the initial point is fixed almost surely at $\gamma^i_\omega(0)=\gamma^i_c(0)$. 
Here, $\gamma_\omega(\tau)$ is referred to as the {\it sample path}. 
The sample path is an element of $\Gamma^0$,  a set of continuous lines between $\gamma^i(0)$ and $\gamma^i(T)$ in which the condition $C^2$ is relaxed to $C^0$ from a definition of the curvilinear path (\ref{cvlin}).
Moreover,  $\gamma_\omega(\tau)$ is unbounded owing to the Brownian motion $B^k(s)$, as will be discussed below.
The random variable $\dot{\gamma}_\omega(\tau)$ must also be treated as a stochastic process.
$B^k(s)~(k=1,2,3)$ are three independent Brownian motions and $\sigma^i_k(s)=\delta^i_k \sigma(s)$
is a square integrable function. 
The Brownian motion $B(s)$ is differentiable nowhere, and the integration measure;
\begin{eqnarray*}
dB^k(s)&=&\frac{dB^k(s)}{ds}ds~=~{\dot B}^kds,
\end{eqnarray*}
must be understood as the Gaussian white noise (Hida-distribution) introduced by Hida\cite{hidasisi2008}.
From {\bf Definition \ref{ip}}, the correspondence is then as $\{\gamma^i_\omega(\tau), \dot{\gamma}^i_\omega(\tau)\} \Leftrightarrow \{X^i_t, \beta^i_s\}$.  
Note that $\dot{\gamma}_\omega$ and $\ddot{\gamma}_\omega$ do {\it not} denote $d\gamma_\omega/d\tau$ and $d\dot{\gamma}_\omega/d\tau$, respectively,  but are random variable independent  of $\gamma_\omega$ and $\dot{\gamma}_\omega$, respectively. 
The stochastic-process induced by the random variable $\gamma_\omega$ defined by (\ref{sdg}) becomes the Brownian motion, which is ensured in the SLM-space following {\bf Remark \ref{brownEM}}. 
Therefore, the phase space $\{\gamma^i_\omega(\tau), \dot{\gamma}^i_\omega(\tau)\}$ can be understood as a pair of random variables representing  a sample trajectory of a mass point on the SLM-space,  while integration (\ref{sdg}) must be understood as a stochastic integration.
The convergence of these integrations are not trivial, unlike the convergence of integration (\ref{fint2}).

The stochastic representation of the canonical momentum $\pi_\omega$ is also introduced as
\begin{eqnarray}
\pi_\omega&=&\frac{\delta\LL}{\delta\dot{\gamma}_c}\Bigl|_{c\rightarrow\omega}.\label{piomega}
\end{eqnarray}
Here, $\bullet|_{c\rightarrow\omega}$ replacing an operation with a corresponding stochastic process (functional variation). 
In the SLM-space, the Lagrangian must be understood to be a random variable defined as $\LL_\omega=\LL(\gamma_\omega,\dot{\gamma}_\omega)(\tau)$; therefore, the action integral can be considered to be a random variable as follows; 
\begin{eqnarray}
\I_\omega=
\int_0^Td\tau~\LL\left(\gamma_\omega(\tau),\dot{\gamma}_\omega(\tau)\right).
\end{eqnarray}
This integration cannot be a Lebesgue-Stieltjes integration because the stochastic process $\LL(\gamma_\omega,\dot{\gamma}_\omega)$ is a continuous but unbounded function.
If it exists, the expected value of this random variable might coincide with the classical action.
The expected value is formally written using (stochastic) integration over all elements of a set of possible sample paths, which are denoted as $\B_\Gamma$.
The existence of the Lebesgue measure on such an infinite dimensional space cannot be expected in general;
a mathematical treatment of the integration on $\B_\Gamma$ will be discussed later.

The classical path is obtained as the expected values of the stochastic process of the $\gamma_\omega$ as follows:
\begin{eqnarray*}
\I_c&=&E_\omega\left[\I_\omega|\B_\Gamma\right]
~=~\sum_{\gamma_\omega\in\Gamma}
\Bigl|\psi(\gamma_\omega)\Bigr|^2\I(\gamma_\omega)
=\int_0^Td\tau~\LL\left(\gamma_c(\tau),\dot{\gamma}_c(\tau)\right),
\end{eqnarray*}
where $\psi(\gamma_\omega)$ is the probability amplitude, which is introduced according to {\it Definition 4.1} in ref.\cite{kurihara2014}.  
Under this setup,  an additional principle---the {\it extremal entropy principle}---can be stated: the quantum probability amplitude and probability measure that extremalize the entropy; 
\begin{eqnarray}
S&=&-\sum_{\gamma_\omega\in\Gamma}\Bigl|\psi(\gamma_\omega)\Bigr|^2
\log{\Bigl|\psi(\gamma_\omega)\Bigr|^2}\label{EntpyEq2},
\end{eqnarray}
under the constraints,
\begin{eqnarray}
\left\{
\begin{array}{l}
\sum_{\gamma_\omega\in\Gamma}
\Bigl|\psi(\gamma_\omega)\Bigr|
^2\I(\gamma_\omega)=\I_c\label{CE5},\\
\sum_{\gamma_\omega\in\Gamma}\Bigl|\psi(\gamma_\omega)\Bigr|^2=1,
\label{PathI1}
\end{array}
\right.
\end{eqnarray}
is given by;
\begin{eqnarray}
\psi(\gamma_\omega)&\simeq&
C~e^{\frac{i}{\hbar}\I(\gamma_\omega)}\label{PathI2},
\end{eqnarray}
where $C\in\C$ is an appropriate normalization constant and $\I$ is rendered dimensionless by dividing by the dimensional constant $\hbar$.
This remark is identical to {\it Theorem 4.1} in ref.\cite{kurihara2014}, whose proof can be found in the same reference.
Here, $\psi$ gives the transition probability (propagator)  from $t=0$ to $t=T$.
At this stage, $\hbar$ is simply a constant to adjust the dimension of the argument of the exponential function.
When it is taken as the Planck constant (divided by $2\pi$), it becomes a solution of  the Shr\"{o}dinger equation\footnote{See, for example, section 1.3 in \cite{kaku1999introduction}}.
The sum over possible paths can be interpreted as the path integral and, therefore, the maximum entropy principle has induced the path integral quantization.

On the other hand, in the context of this report, the sum must be understood as a stochastic integration such as
\begin{eqnarray}
\sum_{\gamma_\omega\in\Gamma}\bullet\Rightarrow
E_\omega
\left[
~\bullet~\Bigl|\B_\Gamma
\right].
\end{eqnarray}
Therefore, the transition amplitude can be obtained from (\ref{PathI2}) as
\begin{eqnarray}
\psi(\gamma(T))&=&C~E_\omega\left[e^{\frac{i}{\hbar}\I(\gamma_\omega)}
\Big|\B_\Gamma\right],\label{Kint}
\end{eqnarray}
where the probability that a mass point moves from $\gamma(0)$ to $\gamma(T)$ is given as $|\psi(\gamma(T))|^2$.
As this integration is not a functional but a stochastic integration, it is not trivial that the  integration (\ref{Kint}) is equivalent to the path integral introduced by Feynman\cite{feynman1942feynman,feynman1965quantum}.
Fortunately, in 1983 Hida and Streit\cite{hida1983,streit1983} rigorously proved that the path integral can be formulated by means of the Hida-distribution for some classes of potential functions.
To do this, they used the following approach\cite{mathematicaltheory}.
The path integral of a test functional $f(\gamma)$ can be expressed formally as
\begin{eqnarray}
I_\mathrm{path}&=&C\int f(\gamma)e^{\frac{i}{\hbar}\I(\gamma)}\DD\gamma,\label{pathI}
\end{eqnarray}
where $\DD\gamma$ is the integration ``measure" of the functional integration, $C$ is an appropriate normalization factor and the integration is performed over all possible paths.
However, as is well known,  the $\DD\gamma$ cannot exist as the Lebesgue measure and, therefore, the integration is not well-defined in general.
To treat this functional ``{\it measure}'' rigorously, they introduced white noise, as is briefly explained  in {\bf Example \ref{WH}} in Appendix \ref{A4-1}.
Using $1=\exp{[\dot{B}^2/2]}\exp{[-\dot{B}^2/2]}$, integration (\ref{pathI}) can be rewritten as;
\begin{eqnarray}
I_\mathrm{path}&=&C\int f(\gamma)e^{\frac{i}{\hbar}\I(\gamma)+\dot{B}(\tau)^2/2}
e^{-\dot{B}(\tau)^2/2}
\DD\gamma,\nonumber\\
&=&C'\int_{E^*} f(\gamma_\omega)e^{\frac{i}{\hbar}\I(\gamma_\omega)+\dot{B}(\tau)^2/2}
\mu_G(d\gamma_\omega),\label{flat}
\end{eqnarray}
where $\dot{B}(\tau)$ is a white noise defined on the dual space $E^*$ of $E=\B_\Gamma$ and  
\begin{eqnarray*}
\mu_G(d\gamma_\omega)&=&e^{-\dot{B}^2/2}\DD\gamma 
\end{eqnarray*}
is the Gaussian measure (classical Wiener measure) on $E^*$.
The flattening factor $\exp{(\dot{B}^2/2)}$ appears in addition to the action integral.
The classical path $\gamma_c(\tau)$ is assumed to be a square-integrable function
\begin{eqnarray*}
\int_0^T|\gamma_c(\tau)|d\tau<\infty,&~&\gamma_c(\tau)\in L_2([0,T],d\gamma) 
\end{eqnarray*}
by means of the norm given by (\ref{dgammac}).
Therefore, the Gel'fand triple $E\subset H=L_2([0,T],d\gamma)\subset E^*$ exists.

Even though the integration measure becomes well-defined as the Hida-distribution, it does not mean this integration converges, and the ``measure'' $\exp({\dot{B}^2/2})\mu_G$ still cannot be understood as the Lebesgue measure.
It is shown that the path integral is convergent when the limit 
\begin{eqnarray}
\lim_{\delta_k\rightarrow0}
\exp{\left[
c\sum_k\left(\frac{\delta_k B}{\delta_k}\right)^2\delta_k
\right]}
\end{eqnarray}
where $c\neq1/2$\cite{hida1983} exists.
There is a wide class of potential functions known to give convergent results\cite{LASCHECK1993221,ANDPANDP19955070107,Venkatesh2001261,deFaria2005}.

From our standpoint that the SLM-space is the most fundamental object in nature, the stochastic-integration representation of the transition probability (\ref{Kint}) has priority over the path integral (\ref{pathI}).
In other words, the path integral approximates the exact representation of transition amplitude and the flattening procedure is not necessary.
In a conclusion, the following remark given. 
\begin{remark}{\bf  (Quantum Amplitude)}{\rm \cite{hida1983}}\\
For the Lagrangian of 
\begin{eqnarray}
\LL&=&\frac{m}{2}\dot{\gamma}^2-V(\gamma),\label{qpaLag}
\end{eqnarray}
the trajectory of the mass point is expected to be {\rm (\ref{sdg})}.
The quantum probability amplitude is given as
\begin{eqnarray}
\psi(\gamma(T))&=&
E_\omega\left[C
\exp{\left[
\frac{i}{\hbar}\frac{m}{2}
\int_0^T\dot{\gamma}_\omega(t)^2dt
\right]}
\exp{\left[
-\frac{i}{\hbar}
\int_0^T V(\gamma_\omega(t))dt
\right]}
\delta\left(\gamma_\omega(T)-\gamma_c(T)\right)
\Big|\B_\Gamma\right],\nonumber\\
&=&
\exp{\left[
\frac{i}{2}\frac{m}{\hbar}
\int_0^T\dot{\gamma}_c(t)^2dt
\right]}
E_\omega\Bigl[C
\exp{\left[
\int_0^T \frac{i}{2}\dot{B}^2
dt\right]}
\exp{\left[
-\frac{i}{\hbar}
\int_0^T V\left(\gamma_c(t)+\sqrt{\frac{\hbar}{m}}\dot{B}\right)dt
\right]}
\nonumber\\&~&\times\sqrt{\frac{m}{\hbar}}
\delta\left(B(t)-\sqrt{\frac{m}{\hbar}}\left(\gamma_c(T)-\gamma_c(t)\right)\right)
\Big|\B_\Gamma\Bigr],\label{qpa2}
\end{eqnarray}
where we set $\dot\gamma\dot{B}dt=0$.
Here, Donscker's $\delta$-functional{\rm \cite{hidasisi2008}}, $\delta(B-a)$, is used.
Instead, the pinned Brownian motion defined in \ref{pbm} gives the same result. 
\end{remark}\noindent
This remark directory follows from the extremal entropy principle and is consistent with that of Hida and Streit\cite{hida1983} except for the omission of the flattening factor.
If the factor in front of $\dot{B}^2$ is replaced as $i/2\rightarrow(1+i)/2$, the representation becomes identical to that in \cite{hida1983} completely (an exact definition of $\dot{B}^2$ is given in Appendix \ref{WNpoly}).
This modification is absorbed in the normalization factor $C$ and will give the same probability amplitude. 
Although the modified form of the flat ``measure" applied to the Gaussian measure might induce some effects on the path integral, it can be confirmed that the path integral with the Gaussian measure gives result consistent with those obtained using the flat measure, as is shown in Appendix \ref{ap7}.
In addition, the integration using the Gaussian measure gives better convergence.
The real part of the factor in front of $\dot{B}^2$ may cause divergent integrations over the path integral; however, using a standard procedure of physics this divergent part can be absorbed into a normalization constant and eliminated by redefining the normalization constant\footnote{
For the standard treatment of the path-integral in physics, see, for example, a section $9.1$ in \cite{opac-b1131978}.
}. 
On the other hand, the definition of the quantum probability amplitude in (\ref{qpa2}) can be expressed as a type of stochastic integration; 
\begin{eqnarray*}
\int g(\gamma_\omega)e^{\frac{i}{\hbar}\I(\gamma_\omega)}
\mu_G(d\gamma_\omega).
\end{eqnarray*}
This integration is simply the Fourier transformation with the Gaussian measure and converges for any functions $g(\gamma_c)$ if the function belongs to the Schwartz space after multiplying the Gaussian weight.

Before closing this subsection, we would like to discuss in further detail a relation between the results in this report and in ref.\cite{kurihara2014} , in which the author treated QM as a phenomenological theory instead of a fundamental theory without assuming any underlying structure, and then constructed the thermodynamics of particle trajectories that were shown to be equivalent to those under QM. 
In this report, the underlying structure of particle dynamics is identified as the probabilistic nature of the SLM-space. 
By contrast, in ref.\cite{kurihara2014} the function $\psi(\gamma)$ is assumed to be a slow-moving function with respect to variance of $\gamma$ around the classical solution.  
This assumption can also be justifies within the framework of this report.
A set of sample paths going through a neighborhood of $\gamma_c(\tau)$, where $\tau$ is a fixed time in $0<\tau<T$, is introduced as;
\begin{eqnarray}
\B_\tau&=&\left\{\gamma_\omega\bigl|\delta_\tau<\epsilon\right\}~\subset~\B_\Gamma,
\end{eqnarray}
where,
\begin{eqnarray}
\delta_\tau&=&|\gamma_c(\tau)-\gamma_\omega(\tau)|
=\sqrt{\sum_{i=1}^3\Bigl|\gamma_c^i(\tau)
-\gamma_\omega^i(\tau)\Bigr|^2},
\end{eqnarray}
and $\epsilon\in\R$.  
The expected value of the function $\psi(\gamma)$, the sample paths of which go through the neighborhood of $\gamma_c$ can be expressed as;
\begin{eqnarray*}
E_\omega\left[\psi(\gamma_\omega)(t)\delta(t-\tau)\bigl|\B_\tau\right]
&=&C\int_{-\epsilon}^{+\epsilon}\psi(\gamma_c)(\tau)
\exp{\left[-\frac{l_\omega^2}{2\sigma^2}\right]}
l^2_\omega dl_\omega,
\end{eqnarray*}
where $l_\omega=|\gamma_c-\gamma_\omega|(\tau)$, $C$ a normalization factor and the integral is understood to be Riemann integration. 
Here, the Gaussian distribution appears as a result of the stochastic integration over the Brownian motion, as shown in equations (\ref{stin1}) and (\ref{stin2}) in Appendix  \ref{A1}.
Therefore, the derivative can be expressed as;
\begin{eqnarray*}
\frac{d\psi(\gamma)(\tau)}{d|\gamma(\tau)|}\Bigg|_{\gamma=\gamma_c}&\sim&
\frac{E_\omega\left[\psi(\gamma_c)(\tau)-\psi(\gamma_{\omega})(\tau)|
\B_\Gamma\right]}{\delta_\tau}
=O(\epsilon)
\end{eqnarray*}
Therefore, the limitation $\lim_{\epsilon\rightarrow 0}d\psi(\gamma_c)/d|\gamma| =0$ can be obtained.

We note that, in this procedure QM  is formulated earlier than classical mechanics.
From this point of view, the reason why the principle of least action can function as the principle for extracting the classical path is that the path obtained from this principle is robust against perturbation form stochastic shaking of the metric. 
This can be understood from the fact that the probability amplitude around the classical path varies very slowly because the first derivative disappears, as shown above, because of the properties of the Brownian motion of the stochastic process.

%%%%%%%%%%%%%%%%%%%%%%%%%%%%%
\subsubsection{Equation of motion}
The discussion in a previous section is based on the white noise analysis of the quantum  partition function, which is equivalent to the path-integral method.
Therefore, the quantum transition function is obtained from the stochastic integration directly instead of through the equation of motion.   
Here, we propose an independent method to introduce a quantum equation of motion based on the stochastic differential equation (a brief explanation of the stochastic differential equation can be found in Appendix  \ref{ap5}).

The Hamiltonian represents the total energy of the system defined on the classical manifold. 
When it does not explicitly include a time-variable, it is conserved during the time evolution of the system. 
This energy conservation must be retained in the SLM-space as a property using the expected value;
\begin{eqnarray}
\frac{d}{dt}
E_\omega\left[\HH(\gamma_\omega,\pi_\omega)|\B_\Gamma\right](t)
&=&
E_\omega\left[d\HH(\gamma_\omega,\pi_\omega)/dt|\B_\Gamma\right]=0.\label{eec}
\end{eqnarray}
Here, the stochastic differentiation of the Hamiltonian can be written as;
\begin{eqnarray}
d\HH(\gamma_\omega,\pi_\omega)&=&
\frac{\delta\HH_c}{\delta\gamma_c}
\Bigg|_{c\rightarrow\omega}\cdot d\gamma_\omega
+
\frac{\delta\HH_c}{\delta\pi_c}
\Bigg|_{c\rightarrow\omega}\cdot d\pi_\omega.\nonumber\\ \label{qmito2}
\end{eqnarray}
We take note that $\pi$ and $\gamma$ are vectors and ``$\cdot$'' denotes the inner product of two three-vectors. 
The functional variation in (\ref{qmito2}) can be obtained as;
\begin{eqnarray}
\frac{\delta\HH_c}{\delta\pi_{ci}}
\Bigg|_{c\rightarrow\omega}&=&{\dot \gamma}^i_\omega,\label{dhdg}
\end{eqnarray}
where $i=1,2,3$ are the components of the spatial coordinate.
The preceding is obtained from the classical relation (\ref{gammadot}) and is {\it not} an equation of motion. 
If solutions to the equation of motion are obtained as $\gamma_\omega(t)$ and $\pi_\omega(t)$, their stochastic differentiations can be respectively written as; 
\begin{eqnarray*}
\left\{
\begin{array}{ccc}
d\gamma^i_\omega&=&\dot{\gamma}^i_c(t)\Bigl|_{c\rightarrow\omega}dt+\sum_{j=1,3}\left[\sigma_\gamma(t)\right]_j^idB^j_\gamma(t),\\
d\pi^i_\omega&=&\dot{\pi}^i_c(t)\Bigl|_{c\rightarrow\omega}dt+\sum_{j=1,3}\left[\sigma_\pi(t)\right]_j^idB^j_\pi(t),
\end{array}
\right.
\end{eqnarray*}
where $\sigma_\bullet(t)$ is a $3\times3$ matrix of functions of $\tau$ and $dB_\bullet(t)$ is  the Gaussian white noise with a mean value of zero and a unit variance. 
When $\dot{\pi}_\omega$ and $\dot{\gamma}_\omega$ are multiplied by the first and second equations, respectively, above, the stochastic differential equations; 
\begin{eqnarray}
\left\{
\begin{array}{ccc}
\dot{\pi}_\omega\cdot d\gamma_\omega
&=&\dot{\pi}_c\cdot\dot{\gamma}_c\Bigl|_{c\rightarrow\omega}dt
+\dot{\pi}_\omega\cdot\sigma_\gamma(t)\cdot dB_\gamma(t),\\
\dot{\gamma}_\omega\cdot d\pi_\omega
&=&\dot{\gamma}_c\cdot\dot{\pi}_c\Bigl|_{c\rightarrow\omega}dt
+\dot{\gamma}_\omega\cdot\sigma_\pi(t)\cdot dB_\pi(t),
\end{array}
\right.
\label{sdegp}
\end{eqnarray}
are obtained, where $\cdot$ indicates either an inner product of two three-vectors or a product of $3\times3$ matrix and three-vector. 
Here, the It\^{o}-rule of $(dt)(dB)=0$ is used, and the relation;
\begin{eqnarray*}
\dot{\gamma}_\omega\cdot d\pi_\omega&=&\dot{\pi}_\omega\cdot d\gamma_\omega+
\dot{\gamma}_\omega\cdot\sigma_\pi(t)\cdot dB_\pi(t)+
\dot{\pi   }_\omega\cdot\sigma_\gamma(t)\cdot dB_\gamma(t)
\end{eqnarray*}
follows from (\ref{sdegp}).
As the Brownian motion is symmetric around its mean value$(=0)$, it can be replaced as $dB\rightarrow-dB$. 
As a result, the stochastic differentiation of the Hamiltonian can be obtained from (\ref{qmito2}) as;
\begin{eqnarray}
d\HH&=&\left(
\frac{\delta\HH_c}{\delta\gamma_c}\Bigl|_{c\rightarrow\omega}
+\dot{\pi}_\omega
\right)\cdot d\gamma_\omega
\nonumber+
\dot{\gamma}_\omega\cdot\sigma_\pi(t)\cdot dB_\pi(t)   +
\dot{\pi   }_\omega\cdot\sigma_\gamma(t)\cdot dB_\gamma(t),\label{dH}
\end{eqnarray}
When the expected value is taken on both sides of equation $(\ref{dH})$, the conservation relation
\begin{eqnarray}
E_\omega\left[d\HH/dt\big|\B_\Gamma\right]&=&
E_\omega\left[\left(
\frac{\delta\HH_c}{\delta\gamma_c}\Big|_{c\rightarrow\omega}
+\dot{\pi}_\omega\right)\cdot \dot{\gamma}_\omega\Big|\B_\Gamma\right]=0,\label{dH2}
\end{eqnarray}
can be obtained because the mean value of the Gaussian noise is zero.
Therefore, the stochastic equations of motion such as
\begin{eqnarray}
d\pi_\omega^i&=&
-\frac{\delta\HH_c}{\delta\gamma_{ci}}\Bigl|_{c\rightarrow\omega}dt
+\sum_{j=1,3}\left[\sigma_\pi(t)\right]_j^idB_\pi^j(t),
\label{sem3}
\end{eqnarray}
follow from (\ref{dH2}). 
After solving the above stochastic equation (of motion), further stochastic integration 
\begin{eqnarray}
d\gamma_\omega^i&=&
\frac{\delta\HH_c}{\delta\pi_c^i}\Bigl|_{c\rightarrow\omega}dt
+\sum_{j=1,3}\left[\sigma_\gamma(t)\right]_j^idB_\gamma^j(t),
\label{sem4}
\end{eqnarray}
from (\ref{dhdg}) can produce a trajectory of the mass point.
These differential equations, including that for Brownian motion, are referred to as the {\it stochastic differential equations} (SDE). 
The existence and uniqueness of the stochastic processes as solutions of the SDEs are ensured mathematically when the stochastic processes satisfy the local Lipschitz condition\cite{hsustochastic}, as is true in this case.

To this point, we have not used the constraint that the Brownian motion is {\it pinned}, as the starting and ending points are fixed as in equation (\ref{pbm}). 
The set of SDEs above does not have invariance under time reversal; to preserve the time-reversal symmetry, another set of SDEs must be added  as follows;
\begin{eqnarray}
\left\{
\begin{array}{ccc}
d\gamma_{\bar{\omega}}^i&=&
-\frac{\delta\HH_c}{\delta\pi_c^i}\Bigl|_{c\rightarrow{\bar{\omega}}}dt
+\left[\sigma_{\bar{\gamma}}(t)\right]_j^idB_{\bar{\gamma}}^{j}(t),\\
d\pi_{\bar{\gamma}}^i&=&
~~\frac{\delta\HH_c}{\delta\gamma_c^i}\Bigl|_{c\rightarrow{\bar{\omega}}}dt
+\left[\sigma_{\bar{\pi}}(t)\right]_j^idB_{\bar{\pi}}^{j}(t).
\end{array}
\right.
\label{sem5}
\end{eqnarray}
The sample path indexed by $\bar{\omega}$ is also an element of $\Gamma^0$.
The stochastic process  $\gamma_{\bar{\omega}}(t)$ with Brownian motion $B_{\bar{\bullet}}(t)$ can be interpreted as a process evolving from the future to the past, as it obey the equation of motion with time reversal.

This approach can be compared with a stochastic quantization of Nelson\cite{PhysRev.150.1079}.
When the functions $\sigma_\bullet(t)$ are set to be constants as;
\begin{eqnarray}
\left\{
\begin{array}{ccc}
\left[\sigma_{\gamma}(t)\right]_j^i&=&\sqrt{\hbar/m}~\delta^i_j,\nonumber\\
\left[\sigma_{\pi}(t)\right]_j^i&=&\sqrt{m\hbar}~\delta^i_j,
\end{array}
\right.
\label{sigma}
\end{eqnarray}
equations of  (\ref{sem3}), (\ref{sem4}) and (\ref{sem5}) together reduce to the Newton-Nelson equation.
For instance, the classical Hamiltonian of
$
\HH_c=\pi_c^2/2m+V(\gamma_c)
$  
gives 
\begin{eqnarray*}
-\frac{\delta\HH_c}{\delta\gamma^i_c}&=&-\frac{d}{d\gamma^i_c}V(\gamma_c)~=~F_i(\gamma_c),\\
\frac{\delta\HH_c}{\delta\pi_{ci}}&=&\frac{1}{m}\pi_c^i.
\end{eqnarray*}  
Therefore, the stochastic equation of motion from $d\HH_\omega=0$ gives a result consistent with stochastic quantisation\cite{gliklikh}.

%%%%%%%%%%%%%%%%%%%%%%%%%%%%%%
\subsection{Uncertainty and commutation relations}\label{sec3-3}
As shown above, the dynamics of a mass point in the SLM-space are equivalent to QM under stochastic quantization. 
If the stochastic quantization is mathematically equivalent to canonical quantization, the canonical commutation relations must be derived from the stochastic properties of the system.  
Although the relation between the commutation relation and stochastic processes has been discussed by several previous authors\cite{PhysRevA.40.530,hida1977,Biane2010698}, 
here we will show the derivation of the commutation relation in our formalism. 

First, let us consider the uncertainty relation in terms of stochastic processes.
The uncertainty relation between the solutions of equations of motion (\ref{sem3}) and (\ref{sem5}) can be obtained on the covariance matrix of two Gaussian distributions.
When the equations of motion are assumed to be linear SDEs such as;
\begin{eqnarray}
d\gamma_\omega(t)&=&\left(
\alpha_\gamma(t)+\beta_\gamma(t)\gamma_\omega(t)
\right)dt+\sigma_\gamma dB_\omega,\label{eqg}
\end{eqnarray}
their solutions can be obtained as
\begin{eqnarray}
\begin{array}{ccl}
\gamma_\omega(t)&=&U_\gamma(t)\left(
\gamma_0+\int_0^t\alpha_\gamma(u)U^{-1}_\gamma(u)du
+\sigma_\gamma\int_0^tU^{-1}_\gamma(u)dB_\omega(u)
\right),\\
U_\gamma(t)&=&\exp{\left[
\int_0^t\beta_\gamma(u)du
\right]},
\end{array}
\label{solg}
\end{eqnarray}
as shown in {\bf Example \ref{LSDE}} in Appendix  \ref{ap5}. 
Here, $\{\alpha_\gamma,\beta_\gamma\}$ are integrable functions given by the Hamiltonian. 
The SDEs and their solutions with respect to the canonical momentum are written by simple replacements of $\bullet_\gamma\rightarrow\bullet_\pi$ in ({\ref{eqg}) and (\ref{solg}).
The expected values of these solutions coincide with classical solutions as follows:
\begin{eqnarray*}
%\left\{
%\begin{array}{ccl}
\gamma_c(t)&=&
E_\omega\left[
\gamma_\omega|\B_\Gamma
 \right],\\
\pi_c(t)&=&
E_\omega\left[
\pi_\omega|\B_\pi
 \right].
%\end{array}
%\right.
\end{eqnarray*}
The classical solutions can be obtained as the solutions of a set of classical equations of motion as;
\begin{eqnarray}
\frac{d\gamma_c(t)}{dt}&=&
\alpha_\gamma(t)+\beta_\gamma(t)\gamma_c(t),~~\gamma_c(0)=\gamma_0,\label{eom1}\\
\frac{d\pi_c(t)}{dt}&=&
\alpha_\pi(t)+\beta_\pi(t)\pi_c(t),~~\pi_c(0)=\pi_0.\label{eom2}
\end{eqnarray}
Therefore, the covariance matrix can be obtained as
\begin{eqnarray}
\sigma^2(\gamma,\pi)
&=&E_\omega[(\gamma_\omega-\gamma_c)(\pi_\omega-\pi_c)|\B_\Gamma],\nonumber\\
&=&\sigma_\gamma\sigma_\pi
\int_0^T
U_\gamma(t)\bar{U}^{-1}_\gamma(t)
U_\pi   (t)\bar{U}^{-1}_\pi   (t)dt,
\nonumber\\
&=&\sigma_\gamma\sigma_\pi
\int_0^T\left(
\int_0^t
U_\gamma(t)U^{-1}_\gamma(u)
U_\pi   (t)U^{-1}_\pi   (u)
du\right)dt,
\nonumber\\
&=&C\hbar,\label{ucr1}
\end{eqnarray}
where 
\begin{eqnarray*}
\bar{U}^{-1}_\gamma(t)&=&\int_0^tU^{-1}_\gamma(s)dB_\omega(s),
~~(\gamma\rightarrow\pi).
\end{eqnarray*}
The constant $C$ can be obtained when the equation of motion and duration $T$ are fixed.
Here we use the It\^{o} rule, $E_\omega[dB_\omega(t)dB_\omega(s)]\Rightarrow \delta(t-s)ds$, given in Appendix  \ref{A1}, and note that $\gamma_\omega$ and $\pi_\omega$ are induced by common stochastic process. 
If these processes are independent each other, then it must be true that $\sigma^2(\gamma,\pi)=0$.
If the functions $\beta_\bullet(t)$ are simply constants, then $C=T^2/2$, and therefore the relation,
\begin{eqnarray}
\frac{1}{T^2}\sigma^2(\gamma,\pi)&=&\frac{\hbar}{2},
\end{eqnarray}
can be obtained, which is simply  Heisenberg's uncertainty relation.

Next, we consider the commutation relation using the same framework and assumptions as above. 
The commutation relation becomes
\begin{eqnarray}
\left[\pi_\omega,\gamma_\omega\right](t)&=&\left(
\pi_\omega\otimes\gamma_\omega-\gamma_\omega\otimes\pi_\omega
\right)(t),\nonumber\\
&=&
\sigma_\pi\sigma_\gamma\Bigl\{
U_\pi(t)\int_0^tU^{-1}_\pi(u)U_\gamma(u)
\left(
\int_0^uU^{-1}_\gamma(v)dB_\omega(v)
\right)dB_\omega(u)-(\pi\leftrightarrow\gamma)
\Bigr\}.\label{cr1}
\end{eqnarray}
Here, the convolution of two stochastic integrations is defined as
\begin{eqnarray*}
X^{12}_\omega(t)&=&
\left(X^1_\omega\otimes X^2_\omega\right)(t)=\int_0^t x^1(u)
\left(\int_0^u x^2(v)dB_\omega(v)\right)
dB_\omega(u),\\
X^i_\omega(t)&=&\int_0^t x^i(u)dB_\omega(u).
\end{eqnarray*}
The convolution induces a new stochastic process $X^{12}_\omega(t)$, and the stochastic integration appearing in equation $(\ref{cr1})$ can be estimated as 
\begin{eqnarray}
\int_0^u
U^{-1}_\gamma(v)
dB_\omega(v)&=&
B_{\omega}(u;\sigma^2),
\end{eqnarray}
where $B_{\omega}(u;\sigma^2)$ is a Brownian motion with zero mean and a variance of 
\begin{eqnarray}
\sigma^2&=&\int_0^u
\left(U^{-1}_\gamma(v)\right)^2
dv,
\end{eqnarray}
as defined in Appendix  \ref{ap4}.
It  is known that two Brownian motions with different variances will be mutually disjoint, and therefore, the stochastic integration vanishes as
\begin{eqnarray}
\int_0^t
U^{-1}_\pi(u)U_\gamma(u)
B_\omega(u;\sigma^2)
dB_\omega(u)&=&0.
\end{eqnarray}
Therefore, any commutation relation comprising two solutions of the equations of motion, that share the same sample path, $\omega$, will be zero due owing to the same mechanism.
Only commutation relations of type $[\pi_{\omega^\star},\gamma_\omega]$ can be non-zero; in this case, a convolution of two stochastic processes can be defined as;
\begin{eqnarray*}
X^{12}_{\omega^\star\omega}(t)&=&
\left(X^1_{\omega^\star}\otimes X^2_\omega\right)(t)
=\int_0^t
\left[\int_t^u x^1(v)dB_{\omega^\star}(v)\right]
\left[\int_0^u x^2(v)dB_\omega(v)\right]
du,\\
&=&\int_0^t
\left[\int_0^{t-u} x^1(v+u)dB_\omega(v+u)\right]
\left[\int_0^u x^2(v)dB_\omega(v)\right]
du,
\end{eqnarray*}
where $dB_{\omega^\star}$ is a time-reversed Gaussian process with the same sample path, $\omega$. 
Each stochastic integral may produce a Gaussian process with a different variance, for instance, $\sigma^2_1$, and $\sigma^2_2$, and as the result new stochastic process $X^{12}_{\omega^\star\omega}$ will also be a Gaussian process with a mean value of zero.
If the integration region is from $-\infty$ to $+\infty$ instead of $0$ to $t$, the variance becomes $\sigma^2_{12}=\sigma^2_1+\sigma^2_2$.
Again, using the same equations (\ref{eom1}) and (\ref{eom2}) a commutation relation can be obtained as;
\begin{eqnarray}
\left[\pi_{\omega^\star},\gamma_\omega\right](t)
&=&\left(
\pi_{\omega^\star}\otimes\gamma_\omega-
\gamma_\omega\otimes\pi_{\omega^\star}
\right)(t),\nonumber\\
&=&
\sigma_\pi\sigma_\gamma
e^{(\beta_\pi+\beta_\gamma)t}
\Bigg\{
\int_0^t\left[
\int_0^{t-u}e^{-\beta_\pi (v+u)}dB_{\omega}(v+u)
\right]\left[
\int_0^ue^{-\beta_\gamma v}dB_{\omega}(v)\right]
du-(\pi\leftrightarrow\gamma)
\Bigg\}.~~\label{CCR}
\end{eqnarray}
This is a convolution of two Gaussian processes, whose variances are
\begin{eqnarray*}
\sigma_\pi^2        &=&\int_u^te^{-2\beta_\pi    v}dv
=\frac{e^{-2u \beta_\pi}-e^{-2t \beta_\pi}}{2\beta_\pi},\\
\sigma_\gamma^2&=&\int_0^ue^{-2\beta_\gamma v}dv
=\frac{1-e^{-2u \beta_\gamma}}{2\beta_\gamma},
\end{eqnarray*}
respectively. 
The variance of the resulting Gaussian process can be obtained as;
\begin{eqnarray*}
\sigma_{\pi\gamma}^2&=\sigma_\pi^2+\sigma_\gamma^2.
\end{eqnarray*}
When $\beta_\pi\neq\beta_\gamma$, it is easily confirmed that $\sigma_{\pi\gamma}^2\neq\sigma_{\gamma\pi}^2$, where 
\begin{eqnarray*}
\sigma_{\gamma\pi}^2&=&
\int_0^ue^{-2\beta_\pi v}dv+
\int_u^te^{-2\beta_\gamma    v}dv.
\end{eqnarray*}
Therefore, owing to the reproducing property of the Gaussian distribution, the commutation relation is the following stochastic processes;
\begin{eqnarray*}
\left[\pi_{\omega^\star},\gamma_\omega\right](t)&=&
\sigma_\pi\sigma_\gamma e^{(\beta_\pi+\beta_\gamma)t}
\left(
B_\omega(t;\sigma_{\pi\gamma}^2)-
B_\omega(t;\sigma_{\gamma\pi}^2)
\right),\\
&=&
\sigma_\pi\sigma_\gamma e^{(\beta_\pi+\beta_\gamma)t}
B_\omega(t;\sigma_{\pi\gamma}^2+\sigma_{\gamma\pi}^2).
\end{eqnarray*}
After taking an average over a long time interval, the variance becomes
\begin{eqnarray}
{\bar\sigma}^2&=&\lim_{t\rightarrow\infty}
\frac{\sigma_{\pi\gamma}^2+\sigma_{\gamma\pi}^2}{t}=\frac{1}{2}
\left|\frac{1}{\beta_\pi}+\frac{1}{\beta_\gamma}\right|.
\end{eqnarray}
Thus, the commutation relation can be expressed as a single Gaussian process with zero mean and the variance derived above. 
The classical limit of this commutation relation is given as;
\begin{eqnarray}
E_\omega\left[
\int_0^T\left[\pi_{\omega^\star},\gamma_\omega\right](t)dt
\Bigg|\B_\omega\right]&=&0.
\end{eqnarray}
Note the similarity between the commutation relation in (\ref{CCR}) and the stochastic area defined in {\bf Definition \ref{STA}}.
The commutation relation can be understood to be the stochastic area on the phase space.
As the average value of the stochastic area on the phase space is zero, the area surrounded by the classical trajectory is conserved (Liouville's theorem).
On the other hand, the variance of the stochastic area is not zero, which means that it cannot vanish as a result of quantum (stochastic) effects (the uncertainty principle). 

Here, we will discuss the relation between the path integral and stochastic equation of motion methods.
In the path-integral method, the meaning of the amplitude, $\psi(\gamma)$, is given in its definition. 
By contrast, the meaning of the solutions of the SDE under the stochastic equation of motion needs clarifying.
The physical meaning of the solutions can be elucidated by investigating the relation between two methods.
First, let us consider an operator such as;
\begin{eqnarray*}
-i\hbar\frac{\delta\psi(\gamma_c)}{\delta\gamma_c}
\Bigg|_{c\rightarrow\omega}
&=&-
\left(
\int_0^T
\frac{\delta\HH_c}{\delta\gamma_c}dt
\right)\psi(\gamma_c)\Bigg|_{c\rightarrow\omega},\\
&=&\left(
\int_0^T
\dot{\pi}_cdt
\right)\psi(\gamma_c)\Bigg|_{c\rightarrow\omega}
=\pi_\omega\psi(\gamma_\omega).
\end{eqnarray*}
This operator can be interpreted as the momentum as follows;
\begin{eqnarray}
-i\hbar\frac{\delta~}{\delta\gamma_c}\Bigg|_{c\rightarrow\omega}&\Rightarrow&\pi_\omega.\label{momop}
\end{eqnarray}
Next, let us look closely at equation $(\ref{PathI2})$ from the point of view of the stochastic equation of motion.
After solving the SDE,  the transition probability (propagator) becomes;
\begin{eqnarray*}
\psi(\gamma(T))&=&C
\int_{\gamma(0)}^{\gamma(T)}
\exp{\left[
\frac{i}{\hbar}\int_0^T
\LL(\gamma_\omega,\dot{\gamma}_\omega)(\tau)d\tau
\right]}\mu_G(d\gamma_\omega).
\end{eqnarray*}
From equations $(\ref{sem3})$ and $(\ref{sigma})$, the solution of the quantum equation of motion is expected to have the form;
\begin{eqnarray}
\gamma_\omega&=&\gamma_c+\sqrt{\frac{\hbar}{m}}B_\omega,
\end{eqnarray}
where $B_\omega$ is the pinned-Brownian motion given in $(\ref{pbm})$.
This stochastic integration has been proven to be equivalent to the Feynman path-integration\cite{hida1983,Kuo1994} (see also chapter 10.4.2 in \cite{albeverio2008}).
The above results show that the amplitude defined in the path-integral method is consistent with solutions based on the quantum equation of motion.

%%%%%%%%%%%%%%%%%%%%%%%%%%%%%%%%
% Section 4: 
%%%%%%%%%%%%%%%%%%%%%%%%%%%%%%%%
\section{ General relativity on the SM-space}\label{sec4}
As SMQ is a method for treating the metric tensor itself, it must be closely related to general relativity.
One might expect to obtain the quantum effects of gravity by treating the space time manifold as the SM-space, as well investigated in this section, in which Planck units setting $4\pi G=c=\hbar=1$ will be used (even if the constant $G$ is written in a formula, its value will be unity). 
\subsection{Geometrical framework}\label{sec4-1}
First, we summarize classical general relativity in terms of a vierbein formalism. 
The mathematical setups introduced in this subsection are based on \cite{Kurihara:2016nap}; only classical (non-probabilistic) setups of the differential geometry are treated here.

First, classical general relativity is geometrically re-formulated in terms of a vierbein formalism. 
Let us introduce a four-dimensional Riemannian manifold $\MM$ with a metric tensor $g_{\bullet\bullet}$ of dimension four, which is referred to as the {\it global manifold}. 
A line element on $\MM$ can be expressed as;
\begin{eqnarray}
ds^2&=&g_{\mu_1\mu_2}(x)dx^{\mu_1}\otimes dx^{\mu_2}\label{LE}.
\end{eqnarray}
The covariant derivative for a tensor defined on $\MM$ can be introduced using the affine connection $\Gamma^\bullet_{~\bullet\bullet}$.
Because the affine connection is not a tensor, it is always possible to find a frame in which the affine connection vanishes as $\Gamma^\bullet_{~\bullet\bullet}(x)=0$ at any point $x$ on the global manifold.  
A local manifold with a vanishing affine connection is required to have the symmetries $SO(1,3)$ and $T^4$ ($4$-dimensional translation symmetry) and is referred to as  the {\it local (Lorenz) manifold} and denoted by $\M$. 
On $\M$,  a vanishing gravity, $\partial_\bullet g_\bullett=0$, is not required, and a transformation from the global to the local manifold can be achieved using a tensor function\footnote{More exactly, simultaneously a rank-one global and a rank-one local vector.} $\Varepsilon_\bullet^\bullet(x)$, such as
\begin{eqnarray*}
g_{\mu_1\mu_2}(x)~\Varepsilon^{\mu_1}_{a}(x)~\Varepsilon^{\mu_2}_{b}(x)
&=&\eta_{ab},
\end{eqnarray*}
where $\eta_{ab}$ is a local Lorentz metric under a particle physics convention such as $\eta_{\bullet\bullet}=\mathrm{diag}(1, -1, -1, -1)$.
Here, we employ the convention that indices of Greek and Roman letters represent coordinates on the global and local manifolds, respectively.
Both indices run from zero to three.
Note that the sign of the determinant of the metric tensor is invariant under the general coordinate transformation. 
The inverse function is represented as $(\Varepsilon^{\mu}_{a})^{-1}=\Varepsilon_{\mu}^{a}$, which satisfies
\begin{eqnarray}
\Varepsilon^{\mu}_{a}(x)\Varepsilon_\mu^{b}(x)=\delta^b_a,
~~\Varepsilon^\mu_{a}(x)\Varepsilon_{\nu}^a(x)=\delta^\mu_\nu.\label{LE2}
\end{eqnarray}
Both $\Varepsilon^{\mu}_{a}$ and its inverse $\Varepsilon_{\mu}^{a}$ are called the {\it vierbein}.
An alternative definition of the vierbein can be given using a set of bi-local functions $\xi^a(y,x)$, defined on a neighborhood $U\in\MM$, in which $x,y\subset U$. 
While $\xi^a$ is a function on the global manifold, it is a vector on the local manifold; it is given as a solution of the differential equation
\begin{eqnarray}
\Varepsilon^a_\mu&=&\frac{\partial\xi^a(y,x)}{\partial y^\mu}\Bigl|_{y=x},\label{exi}
\end{eqnarray}
and of equations (\ref{LE}) and (\ref{LE2}) with the boundary condition $\xi^a(x,x)=0$.
The {\it vierbein (one-)form} is introduced as $\eee^a=\Varepsilon^a_\mu dx^\mu$ and behaves as a local vector under the Lorentz transformation (throughout this report, Fraktur letters indicate differential forms) .

The covariant (under the general coordinate transformation) derivative $\nabla_\mu$ applying on a local tensor can be written using the spin-connection $\omega$ as follows:
\begin{eqnarray*}
\nabla_\mu T^a_{~b}&=&
\partial_\mu T^a_{~b}+
\omega^{~a}_{\mu~c}T^c_{~b}-
\omega^{~c}_{\mu~b}T^a_{~c},
\end{eqnarray*}
where
\begin{eqnarray*}
\omega^{~a}_{\mu~b}&=&
\Varepsilon^a_\nu
\Gamma^\nu_{~\mu\rho}~\Varepsilon^\rho_b-
\left(\partial_\mu\Varepsilon^a_\rho\right)\Varepsilon^\rho_b.
\end{eqnarray*}
A {\it spin-connection form } (or simply {\it spin form}) can be defined using the spin connection as
\begin{eqnarray*}
\vomega^{ab}&=&dx^\mu\omega^{~a}_{\mu~c}~\eta^{cb}
=\vomega^{a}_{~c}~\eta^{cb},
\end{eqnarray*}
which satisfies $\vomega^{ab}=-\vomega^{ba}$. 
We note that the spin form is a generator of the Lie algebra of $\mathfrak{so}(1,3)$ and is {\it not} a Lorentz tensor.
The {\it torsion}  and {\it curvature forms} can be introduced as;
\begin{eqnarray*}
\TTT^a&=&d\eee^a+
\vomega^{a}_{~b}\wedge\eee^{b},\\
\RRR^{ab}
&=&d\vomega^{ab}+
\vomega^{a}_{~c}\wedge\vomega^{cb},
\end{eqnarray*}
respectively.
Any two-forms defined on the local manifold can be expanded using the orthogonal bases of the two-forms on the local manifold,  $\eee^\bullet\wedge\eee^\bullet$. 
The curvature form can be expressed as
\begin{eqnarray}
\RRR^{ab}&=&\frac{1}{2}\Ri^{ab}_{~~c_1c_2}~\eee^{c_1}\wedge\eee^{c_2}.\label{curv1}
\end{eqnarray}
The expansion coefficients $\Ri^{\bullet\bullet}_{~~\bullet\bullet}$ are referred to the {\it Riemann curvature tensor}.

The volume form can be expressed using vierbein forms, e.g.
\begin{eqnarray}
\vvv&=&\frac{1}{4!}\epsilon_{a_1a_2a_3a_4}
\eee^{a_1}\wedge\eee^{a_2}\wedge\eee^{a_3}\wedge\eee^{a_4},\label{VF}
\end{eqnarray}
where $\epsilon_{\bullet\bullet\bullet\bullet}$ is a complete anti-symmetric tensor with $\epsilon_{0123}=+1$.
Similarly, three-dimensional volume- and surface-forms are introduced as;
\begin{eqnarray*}
\VVV_{a}&=&\frac{1}{3!}\epsilon_{ab_1b_1b_2}\eee^{b_1}\wedge\eee^{b_2}\wedge\eee^{b_3},\\
\SSS_{ab}&=&\frac{1}{2!}\epsilon_{abc_1c_2}\eee^{c_1}\wedge\eee^{c_2}.
\end{eqnarray*}
The surface form $\SSS_{ab}$ is perpendicular to plane spanned by $\eee^a$ and $\eee^b$. 
Hereafter, dummy {\bf Roman} indices are omitted in expressions when the index pairing is obvious, with dots placed instead at the index, e.g. $2\SSS_{ab}=\epsilon_{ab\cdott}\eee^{\bm\cdot}\wedge\eee^{\bm\cdot}$.

The Lagrangian form and action integral for gravitation can be represented in this terminology as;
\begin{eqnarray}
\LLL_G&=&\frac{1}{2!}\RRR^\cdott(\vomega)
\wedge\SSS_\cdott
-\Lambda~\vvv,\label{Lagrangian44}
\end{eqnarray}
and
\begin{eqnarray*}
\I_G
&=&\int_{\Sigma}\LLL_G,
\end{eqnarray*}
respectively, where $\Lambda$ is the cosmological constant.
To complete the general relativity formulation, Lagrangian forms of the matter and gauge fields must be added. 
The equation of motion can be obtained as the Euler--Lagrange equation by taking variations with respect to $\vomega^{\bullet\bullet}$ and $\SSS_{\bullet\bullet}$.
The Einstein equation and the torsion-less condition can be obtained as the Euler-Lagrange equation of motion as;
\begin{eqnarray*}
\epsilon_{a\cdott\bm\cdot}\left(
\frac{1}{2}\RRR^\cdott\wedge\eee^{\bm\cdot}
-\frac{\Lambda}{3!}\eee^{\bm\cdot}\wedge\eee^{\bm\cdot}\wedge\eee^{\bm\cdot}
\right)
&=&\frac{\kappa_E}{2}T_{a\bm\cdot}\VVV^{\bm\cdot},
\end{eqnarray*}
and $\TTT^\bullet=0$, respectively, where $T_{\bullet\bullet}$ is the energy-momentum tensor of the matter and gauge fields.

%%%%%%%%%%%
\subsection{Construction of local SLM-space}\label{sec4-2}
The general relativity equipping the local SLM-space is constructed in this sub-section. 
The stochastic aspects of general relativity may be implemented by modifying the vierbein form; through which the effects of these aspect may appear on the spin and surface forms.

Let us start from the definition of the vierbein in (\ref{exi}).
The bi-local function $\xi^a(y,x)$ has non-local information of the space time manifold, and the vierbein can be defined through the differentiation of this function. 
By contrast, the distance between two points on the SLM-space is a random variable, that follows a Gaussian distribution with a mean value and variance given by the geometrical distance between the points, as per {\bf Remark \ref{brownEM}}.
A possible modification of the function $\xi^a(y,x)$ that is consistent with the above requirements is given as follows;
\begin{eqnarray}
\hat{\xi}^a(y,x)&=&\xi^a(y,x)+\sigma_\varepsilon\int^ydB^{(a)}(x),\label{xiB}
\end{eqnarray}
where $B^{(a)}(x)~(a=0,1,2,3,)$ are four independent Brownian motions.
Note that the second term of the right-hand side of (\ref{xiB}) represents the total fluctuation of the length between $x$ and $y$ owing to the Brownian motion; a constant, $\sigma_\epsilon$, is prepared to adjust the variance as well as the dimension.
The order parameter of the Brownian motion is not simply a single parameter but a four-dimensional coordinate-vector on the local coordinate-patch of the global manifold.
Moreover, $dB^{(a)}(x)$  must be understood as the white noise introduced in Appendix  \ref{A4-2}.
In this definition, the stochastic function still retains the condition $\hat{\xi}^a(x,x)=0$.
Note that $B^{(a)}(x)$ cannot be considered to be a vector on $\M$; if it is a Lorentz vector with each component fluctuating independently on some local inertial frame, the components will mix with each other following a Lorentz transformation and therefore will not be independent.
In fact, white noise can be treated as an independent Gaussian process in any local coordinate systems (detailed discussions on the $SO(1,3)$ symmetries of white noise are given in Appendix \ref{ap8}.)
From the stochastic function $\hat{\xi}^a$, the stochastic vierbein is obtained as;
\begin{eqnarray*}
\hat\Varepsilon^a_\mu(x)&=&\frac{\partial\xi^a(y,x)}{\partial y^\mu}\Bigl|_{y=x}+\sigma_\varepsilon
\frac{\partial B^{(a)}(x)}{\partial x^\mu},
\end{eqnarray*}
following (\ref{exi}).
Here, the derivative of the Brownian motion is merely a formal expression, as it is differentiable nowhere; an exact definition will be introduced later.
As a short hand notation,  the Brownian motion is denoted as $B_{;\mu}^{(a)}(x)=\partial B^{(a)}(x)/\partial x^\mu$ hereafter.

In the SLM-space, the vierbein form can be a stochastic process formally expressed as
\begin{eqnarray}
\hat{\eee}^{a}(x)&=&
\left(\eee^{a}+
\sigma_{\varepsilon}{\BBB}^{(a)}\right)(x)=\left(
\Varepsilon^a_\mu(x)+\sigma_\varepsilon B_{;\mu}^{(a)}(x)
\right)dx^\mu,\label{stvirbein}
\end{eqnarray}
where ${\BBB}^{(a)}(x)=B_{;\mu}^{(a)}(x)dx^\mu$ ($a=0,1,2,3$) can be understood as the local one-form valued white-noise.
The $\sigma_\epsilon B_{;\bullet}^{(a)}$ must have the same dimensionality as the vierbein, which is dimensionless; therefore, $\sigma_\epsilon$ has the dimension of length only.

The precise definition of $B_{;\bullet}(x)$ can be stated as follows.
Consider the probabilistic distribution introduced in Appendix \ref{ap4}, which will be denoted as $B_{;\bullet}(x)$, under the following setups.
A probability space, $(E^*, \B, \mu)$, is introduced on the square-integrable functions $H=L^2(\R^4)$ defined on the global manifold, the Gel'fand triple $E\subset H\subset E^*$ and the sigma-field $\B$ generated by the cylinder subset of $E^*$.
\begin{definition}\label{VWN}{\bf Vierbein White-Noise (VWN)}\\
The ${B}_{;\mu}(x)$ is called the vierbein white-noise, when its characteristic functional is
\begin{eqnarray*}
C^\mu(x)&=&
\exp{\left[
-\frac{\parallel x\parallel^2}{2}
\right]}e^\mu
=\exp{\left[
-\frac{\d(x)^2}{2}
\right]}e^\mu,
\end{eqnarray*}
where $e^\mu$ is the orthonormal base  on the local coordinate-patch and $x^\mu$ is a coordinate vector. $\d(x)$ is the geometrical distance defined in {\bf Definition \ref{PMM}}.
\end{definition}\noindent
This characteristic functional induces the Gaussian measure uniquely as
\begin{eqnarray*}
C^\lambda(\mu)
&=&\int_{E^*}\exp{\left[i
{B}_{;\mu}x^\mu
\right]}d\mu(x^\lambda).
\end{eqnarray*}
The existence of this measure is ensured by the Bochner--Minlos theorem\cite{hida1980brownian,hidasisi2008}.
The VWN is a stochastic functional with a Gaussian distribution and is pointwise independent.
Note that, in the definition of the stochastic vierbein $\hat{\eee}$ from equation (\ref{stvirbein}) and {\bf Definition \ref{VWN}}, the Brownian {\it motion} itself does not appear explicitly.
In the VWN, the one-dimensional order parameter, the Morkovian time, does not exist; therefore, the stochastic processes do not allow for simple interpretation such as a sample path of the one-dimensional trajectory of the mass point.
One interpretation of the VWN, for instance, is that it gives a {\it sample  universe}, that equips the sample metric tensor.

To confirm that the VWN can construct the SLM-space properly, let us calculate the line element using the stochastic vierbein form equation (\ref{stvirbein}):
\begin{eqnarray*}
d\hat{s}^2&=&\eta_{\cdott}~\hat{\eee}^{\bm\cdot}\otimes\hat{\eee}^{\bm\cdot}
=ds^2+\eta_{\bm\cdott}
\left(2\sigma_\varepsilon\Varepsilon_{\mu_1}^{\bm\cdot}{B}_{;\mu_2}^{(\bm\cdot)}
+\sigma_\varepsilon^2{B}_{;\mu_1}^{(\bm\cdot)}{B}_{;\mu_2}^{(\bm\cdot)}
\right)dx^{\mu_1}\otimes dx^{\mu_2}.
\end{eqnarray*}
The terms in parentheses in the above equation can be understood as stochastic processes of the Gaussian type.
The first term becomes a Gaussian process with a mean value of zero and a variance of;
\begin{eqnarray*}
\sigma_B^2&=&4\sigma_\varepsilon^2~\eta_{\cdott}
\eee^{\cdot}\otimes\eee^{\cdot}\propto\d(|ds|)^2
\end{eqnarray*}
as shown in equation (\ref{stin1}). 
When we set $\sigma_\varepsilon=1/2$, the variance of the white noise becomes $\sigma^2_B=\d(|ds|)^2$.
As mentioned previously, $\sigma_\varepsilon$ should have a length dimension, and therefore the precise representation must be $\sigma_\varepsilon=l_p/2$, where $l_p= \sqrt{\hbar G/c^3}$ is the Planck length.
Although the third term includes Brownian motion,  it is not a stochastic integration but it gives $dB(x)dB(x)=dx$ from the It\^{o} rule.
Therefore, this term gives a fixed correction to the line element, and can be neglected as a higher order effect.
Finally, the stochastic line element can be obtained as;
\begin{eqnarray*}
d\hat{s}^2(x)&=&ds^2\left(
1+B
\right)(x),
\end{eqnarray*}
where $B(x)=\eta_\cdott\eee^{\bm\cdot}(x)\otimes\BBB^{\bm\cdot}(x)$.
This is the representation of the line element on the local SLM-space.

\subsection{Stochastic effects on basic forms}\label{sec4-3}
The volume form (\ref{VF}) on the local SLM-space becomes;
\begin{eqnarray}
\hat{\vvv}&=&\frac{1}{4!}
\epsilon_{a_1a_2a_3a_4}
\hat{\eee}^{a_1}\wedge\hat{\eee}^{a_2}
\wedge\hat{\eee}^{a_3}\wedge\hat{\eee}^{a_4},\nonumber\\
&=&\frac{1}{4!}
\epsilon_{a_1a_2a_3a_4}\huge(
    \eee^{a_1}\wedge\eee^{a_2}\wedge\eee^{a_3}\wedge\eee^{a_4}+
   2{\BBB}^{(a_1)}\wedge\eee^{a_2}\wedge\eee^{a_3}\wedge\eee^{a_4}
\nonumber\\&~&~~~~~~~~~~~~~
  +\frac{3}{2}{\BBB}^{(a_1)}\wedge{\BBB}^{(a_2)}\wedge\eee^{a_3}\wedge\eee^{a_4}
  +\frac{1}{2}{\BBB}^{(a_1)}\wedge{\BBB}^{(a_2)}\wedge{\BBB}^{(a_3)}\wedge\eee^{a_4}
\nonumber\\&~&~~~~~~~~~~~~~~~
    +\frac{1}{16}{\BBB}^{(a_1)}\wedge{\BBB}^{(a_2)}\wedge{\BBB}^{(a_3)}\wedge{\BBB}^{(a_4)}
\huge),~~\label{vf}
\end{eqnarray}
where (\ref{stvirbein}) and  $\sigma_\varepsilon=1/2$ are used.
As all white noises appearing here are independent, the mean value of the volume form becomes the geometric volume as $E[\hat\vvv|\B]=\vvv$, as expected.
On the other hand, the variance of the volume form has a finite value, which corresponds to the zero-point vibration of space time.
This can be interpreted to mean that the four-dimensional volume cannot be zero owing to the uncertainty principle of quantum general relativity.
In cases of a flat global manifold, the vierbein is simply a unit matrix, in which case the variance of the white noise is estimated to be
\begin{eqnarray*}
\delta\left[B_{;\mu}\right]&=&
E\left[\left(B_{;\mu}\right)^2|\B\right]-E\left[B_{;\mu}|\B\right]^2~=~dx^\mu,
\end{eqnarray*}
from the It\^{o} rule.
Thus, we can obtain
$
\delta\left[\BBB^{(a)}\right]=\eee^a
$
for the flat space time case, and the variance of the volume form becomes;
\begin{eqnarray}
\delta_\vvv&=&\frac{65}{16}~\vvv\label{delv}
\end{eqnarray}
form (\ref{vf}).

The stochastic effect on other forms can be estimated similarly.
For both $\VVV_a$ and $\SSS_{ab}$, the mean values are the same as the classical values.
The variances are obtained as;
\begin{eqnarray}
\delta_{\VVV a}~=~\frac{19}{8}\int\VVV_a,&~&\hskip 3mm
\delta_{\SSS ab}~=~\frac{5}{4}\int\SSS_{ab},\label{deltvs}
\end{eqnarray}
for the flat global spacetime case.
In general, the variance of a $n$-dimensional volume form in the flat SLM-space is given as $\delta_n=(3/2)^n-1$.

%%%%%%
\subsection{Stochastic effect on the universe}\label{sec4-4}
The variance of the volume form (\ref{delv}) has the same shape as the cosmological term in the Lagrangian form (\ref{Lagrangian44}).
Therefore, even starting from a universe with $\Lambda=0$, an effective cosmological term appears because of the fluctuation of the metric.
To discuss the quantum effect on the universe precisely, it is necessary to formulate the quantum Einstein equation in the stochastic metric space and solve it under certain boundary conditions.
As these tasks are beyond the scope of this report, we employ a semi-classical approximation of quantum general relativity to estimate the quantum effect from the fluctuations of forms.
In this approximation, the effect of the stochastic metric on the vierbein and other forms are considered around the classical solution of the (classical) Einstein equation.

Let us estimate an order of total fluctuation for a flat universe under this approximation.
To estimate a four-dimensional volume of the current universe, $\int\vvv$, we need to know a complete history of the universe.
Here, the simple approximation of a constant expansion of the universe is assumed.
A recent precise measurement gives an age of the universe to be $t_0=(4.35\pm0.01)\times10^{17}$ sec\cite{Ade:2015xua,Olive:2016xmw}. 
The total four-dimensional volume with a radius of $t_0$ can be calculated as
\begin{eqnarray*}
V_U=\int_{Univ}\vvv~=~\frac{\pi^2}{2}t_0^4,
\end{eqnarray*}
for a flat, constantly expanding universe (this is very rough estimation. It is known that the a comoving distance of the universe is about three times larger than $c\times t_0$).
Thus, the variance becomes;
\begin{eqnarray*}
\delta_\vvv&=&\frac{65}{16}V_U
~=~3.26\times10^{167}~\mathrm{(GeV)}^{-4}.
\end{eqnarray*}
The standard deviation averaged over the universe is
\begin{eqnarray*}
\bar{\sigma}_\vvv&=&\frac{\delta_\vvv^{1/2}}{V_U}~=~2.13\times10^{-84}~\mathrm{(GeV)}^{2}.
\end{eqnarray*}
Interestingly, recent measurements of the dark energy fives\cite{Ade:2015xua,Olive:2016xmw}; 
\begin{eqnarray*}
\Lambda&\sim&3H_0~=~6.16\times10^{-84}~\mathrm{(GeV)}^{2}.
\end{eqnarray*}
This coincidence suggests that the cosmological constant in the classical Hamiltonian of the universe is zero ($\Lambda=0$) and the origin of the dark energy (the cosmological term) may be a quantum fluctuation of space time arising from the stochastic effect.
We note that the unnaturally small values appearing here originated from division by a large number (the total volume of the universe in Planck units).
The fluctuation itself is on the order of $\sigma_\epsilon=l_p/2$,
where the Planck length $l_p$ is explicitly written to emphasize that the $\sigma_\epsilon$ has a length dimension.
Even though the cosmological constant is zero at the classical level, an effective term corresponding to the cosmological term appears towing to the stochastic (quantum) effect as;
\begin{eqnarray}
\sigma_\vvv\vvv&\sim&\sqrt{\frac{65}{16}}V_4^{-1/2}~\vvv,\label{qcc}
\end{eqnarray}
where $V_4$ is the four-dimensional volume in question.
This effective term may therefore have induced the re-acceleration seen in the universe\cite{0004-637X-517-2-565,1538-3881-116-3-1009}.

%%%%%%%%%%%%%%%%%%%%%%%%%%%%%%%%
% Section 5: 
%%%%%%%%%%%%%%%%%%%%%%%%%%%%%%%%
\section{Summary}\label{sec5}
One possible way to explain quantum effects in terms of the stochastic effects of the metric tensor is discussed in this report.
The SM-space can be constructed as a model of the universe in a mathematically legitimate manner.
Besides a mathematical model of a space time manifold, a physical principle is necessary to introduce dynamics into the physical system. 
In this work, the entropy extramization principle is employed instead of the standard principle of least action.
The entropy extramization principle allows to introduce quantum mechanics primarily before formulating classical mechanics.
The entropy (\ref{EntpyEq2})  is defined using the probability amplitude instead of the probability density, and it must be understood as an independent principle from the entropy extramization.
Some characteristics  of quantum mechanics, e.g. quantum entanglement, are due to this principle.

It is shown that the QM of a mass point are equivalent to its classical mechanics on the SLM-space.
The path-integration, quantum equation of motion, uncertainty relation and commutation relations were re-formulated on the SLM-space and shown to be consistent with standard definitions. 
In the formalism proposed in this study, QM is formulated based on the extremal entropy principle a priori and before classical mechanics.
If the path integration described in this report is more fundamental than the standard definition, the path integral is mathematically well-defined on the Gaussian white-noise measure space.

Because the SM-space directly treats the metric  tensor of a manifold, its application to general relativity is straightforward. 
Basic forms appearing in the four-dimensional space time manifold can be treated as random variables and investigated using the white-noise analysis initiated by Hida.
Within the semi-classical approximation, a possible quantum effect on the  dark energy is discussed.
This result might be an explanation of the origin of the small cosmological ``constant'' observed in the current universe. 
However, the question as to why zero-point oscillations of the matter and gauge fields do not generate a large vacuum energy, that is refereed to as the ``cosmological constant problem'', remains unanswered.

To answer this question,  a simultaneous treatment of gravity and matter fields is necessary. 
For the pure gravitational Hamiltonian $\HH_G$ and matter-filed Hamiltonian density $\HH_M$, a total energy can be written as;  
\begin{eqnarray*}
E&=&\int\left(
\HH_G+\HH_M
\right)\vvv.
\end{eqnarray*}
A matter-gravitation interaction is hidden in $\HH_M$ as, e.g. a term $g_{\mu\nu}(\partial^\mu\phi)(\partial^\nu\phi)$in the matter Hamiltonian.
There are two sources of statistical fluctuations in a total energy, the metric tensor in the Hamiltonian and in the volume form.
In the SM-space, a Hamiltonian have a statistical fluctuation $\delta\HH_\bullet$ around a mean value $\HH^0_\bullet$ such as $\HH_\bullet\simeq\HH^0_\bullet+\delta\HH_\bullet$.
The volume form can be expressed also as $\vvv\simeq\vvv_0+\delta\vvv$ as shown in section V.
When only a matter field is considered with ignoring the gravitational contribution, the total energy becomes;
\begin{eqnarray*}
E_M&=&\int\left(\HH_M^0+\delta\HH_M\right)\left(\vvv_0+\delta\vvv\right)
=\int\HH_M^0\vvv_0+\int\delta\HH_M\vvv_0+\int\HH_M\delta\vvv.
\end{eqnarray*}
The first term shows the classical energy and the second therm shows the quantum effect of the matter field.
The third term gives a quantum effect on the matter-gravity interaction.
A fluctuation of the volume form is a order unity $|\delta\vvv|/|\vvv_0|={\cal O}(1)$ as given in (\ref{delv}).
Therefore, a term $\delta\HH_M\vvv$ can not be neglected.
This term is proportional to the total volume and can be absorbed in a definition of the vacuum energy together with the zero-pint oscillation under the approximation ignoring a matter-gravity interaction.
In this report, only the gravitational Hamiltonian is considered with ignoring the matter field in section V.
To treat the cosmological constant problem,  the term $\HH_M\delta\vvv$ must be treated.
Moreover, the quantum effect on the spin form is not discussed in this study.
The clarification of such a treatment of the spin form involves simply clarifying what is a parallel translation of a vector on the SM-space and will be discussed as an independent study elsewhere.

\vskip 5mm
%%%%%%%%%%%%%%%%%%%%%%%%%%%%%%%%%%%%%%%%%%%%%%%%%%%%%%%%%%%%%%%
The author appreciates intensive and critical discussions with Prof. T.~Kaneko, Prof. Yuasa, Prof. J.~Fujimoto.
The author also wishes to thank Dr.~Y.~Sugiyama for his continuous encouragement and fruitful discussions.
%%%
%%%
%
% Appendix
%
%\newpage
\appendix
%%%%%%%%%%%%%%%%%%%%%%%%%%%%%%%%
% Appendix A
%%%%%%%%%%%%%%%%%%%%%%%%%%%%%%%%
\section{Proof of Remark \ref{PMMS}}\label{ap1}
\noindent
{\bf Remark \ref{PMMS}}\\
The stochastic Lorentz metric space is a stochastic pseudometric space\cite{schweizer1960statistical}. 
\begin{proof}
Consider three points, $p,q$ and $r$, and whose distribution functions are denoted as $F_e(s;\d(p-q))$, $F_e(s';\d(q-r))$, and $F_e(s+s';\d(p-r))$ with $s,s'\geq0$. 
These distribution functions follow an exponential distribution as defined in {\bf Definition \ref{PMM}}. 
Then distances among those points must keep the triangle inequality such as;
$
\d(p-q)+\d(q-r)\geq \d(p-r).
$
In this case, the triangle inequality can be deduce as
\begin{eqnarray*}
&~&F_e(s+s';p-r)\geq T(F_e(s;p-q),F_e(s';q-r))\\
&\Rightarrow&
\exp{\left(-\frac{s} {\d(p-q)}\right)}+
\exp{\left(-\frac{s'}{\d(q-r)}\right)}\geq
\exp{\left(-\frac{s+s'}{\d(p-q)+\d(q-r)}\right)}
+\exp{\left(-\frac{s}{\d(p-q)}-\frac{s'}{\d(q-r)}\right)}.
\end{eqnarray*}
From assumptions of the equal-time triangle inequality and $s,s'\geq0$, 
\begin{eqnarray*}
\frac{s+s'}{\d(p-r)}\leq\frac{s+s'}{\d(p-q)+\d(q-r)}
\end{eqnarray*}
follows. Since $\exp{(-x)}$ is a monotonically decreasing function for $x>0$ such as;
\begin{eqnarray*}
\exp{\left(-\frac{s+s'}{\d(p-r)}\right)}\geq
\exp{\left(-\frac{s+s'}{\d(p-q)+\d(q-r)}\right)}
\end{eqnarray*}
can be obtained. Furthermore, since all $d$'s and $s,s'$ are positive,
\begin{eqnarray*}
\mathrm{Max}\left[\frac{s}{\d(p-q)},\frac{s'}{d(q-r)}\right]
\geq\frac{s+s'}{\d(p-q)+\d(q-r)}\geq
\mathrm{Min}\left[\frac{s}{\d(p-q)},\frac{s'}{\d(q-r)}\right],
\end{eqnarray*}
with equality on either side {\it iff} $s/{\d(p-q)}=s'/\d(q-r)$. Therefore, we can obtain as;
\begin{eqnarray*}
Min\left[\frac{s}{\d(p-q)},\frac{s'}{\d(q-r)}\right]\leq\frac{s+s'}{\d(p-r)},
\end{eqnarray*}
then,
\begin{eqnarray*}
Max\left[
\exp{\left(-\frac{s} {\d(p-q)}\right)},
\exp{\left(-\frac{s'}{\d(q-r)}\right)}\right]\geq
\exp{\left(-\frac{s+s'}{\d(p-r)}\right)}.
\end{eqnarray*}
On the other hand, due to $\exp{(-x)}<1$ for $x>0$, one can obtained as;
\begin{eqnarray*}
Min\left[
\exp{\left(-\frac{s} {\d(p-q)}\right)},
\exp{\left(-\frac{s'}{\d(q-r)}\right)}\right]\geq
\exp{\left(-\frac{s}{\d(p-q)}-\frac{s'}{\d(q-r)}\right)}.
\end{eqnarray*}
Therefore,
\begin{eqnarray*}
&~&
\exp{\left(-\frac{s} {\d(p-q)}\right)}+
\exp{\left(-\frac{s'}{\d(q-r)}\right)}\\
&~&=
Max\left[
\exp{\left(-\frac{s} {\d(p-q)}\right)},
\exp{\left(-\frac{s'}{\d(q-r)}\right)}\right]
+Min\left[
\exp{\left(-\frac{s} {\d(p-q)}\right)},
\exp{\left(-\frac{s'}{\d(q-r)}\right)}\right]\\
&~&\geq
\exp{\left(-\frac{s+s'}{\d(p-r)}\right)}+
\exp{\left(-\frac{s}{\d(p-q)}-\frac{s'}{\d(q-r)}\right)},\\
&~&\geq
\exp{\left(-\frac{s+s'}{\d(p-q)+\d(q-r)}\right)}+
\exp{\left(-\frac{s}{\d(p-q)}-\frac{s'}{\d(q-r)}\right)}.
\end{eqnarray*}
Thus, the lemma is proved.
\end{proof}

%%%%%%%%%%%%%%%%%%%%%%%%%%%%%%%%
% Appendix B
%%%%%%%%%%%%%%%%%%%%%%%%%%%%%%%%
\section{Poisson process}\label{ap2}
When a probability density of a random variable $X$  satisfied a relation,  
\begin{eqnarray*}
P(X>s+t|X>t)=P(X>s),
\end{eqnarray*}
for any $0<s,t\in\R$, it  is recognized to have a  {\it memoryless property}.
Here, $P(a|b)$ is a conditional probability, which gives a probability observing $a$ when $b$ is realized. 
It is easily confirmed that the exponential probability density has a memoryless property, and thus the stochastic process induced by the exponential probability density becomes a Markov process, in which the probability to a next step only depends on the current status. 
The exponential probability density can be induced from the Poisson process.
Considering that a series of independent random variables $\{X_i\}_{i=1,2, \cdots}$ whose probability measure with $P(0<s<X_i<s+ds)=p(s)ds$.
A series of random variables $\{A_j\}_{j=0,1,2,\cdots}$ such as;
\begin{eqnarray*}
A_0=0~~a.s.,&~&A_j=\sum_{i=1}^jX_i
\end{eqnarray*}
is referred as to the ``{\it stochastic process induced by the probability measure$(p(s)ds)$}''.
Another random variable $0\leq N(s)\in\N$ is also introduced as
$
N(s)=\mathrm{max}\left\{k;A_k\leq s\right\}.
$
A frequency distribution is defined for any $n$ and $0=s_0<s_1<\cdots<s_n$ as $\Delta N_i=N(s_i)-N(s_{i-1})$. 
This system can be interpreted as a stochastic process as follows:
\begin{itemize}
\item $s$ is a ``{\it length}'' variable,
\item $X_i$ is a moving distance at an $i$'th step.
\item $p(s)ds$ is probability to move with a distance of $s$.
\item $A_i$ is a total moving distance up to the $i$'th step.
\item $N(s)$ is a total number of steps to reach a position $s$.
\end{itemize}
If $N(s)$ is considered as a fictitious ``{\it time}'', random variables $s/N(s)$ can be interpret as an {\it average velocity} and $\Delta s_i/\delta N_i$ an {\it instantaneous velocity}, where $\Delta s_i=s_i-s_{i-1}$.
When the exponential measure $\mu_e(s;\lambda)$ is used as the  probability measure, the process becomes the Poisson process.
A definition of the Poisson process is given as follows:
\begin{definition}{\bf Poisson process}\label{pp}\\
When a random variable $\{0\leq N(s)\in\Z\}$ for $0\leq s\in\R$ defined on the probability space $(\Omega,\ff,P)$ is satisfied following conditions, it is called ``{\it Poisson process}'';
\begin{enumerate}
\item $N(0)=0$~~a.s.,
\item For $\forall\omega\in\Omega$, $N_\omega(s)$ is left-continuous and monotonically increasing,
\item The frequency distribution of the random variable $N(s)$ makes the Poisson distribution as
\begin{eqnarray*}
P(\Delta N_i;\Delta s_i/\lambda)&=&%\prod_{i=1}^n
\frac{\left(\Delta s_i/\lambda\right)^{\Delta N_i}}{\Delta N_i!}
\exp{\left(-\Delta s_i/\lambda\right)},
\end{eqnarray*}
where 
\begin{eqnarray*}
\Delta N_i&=&N(s_i)-N(s_{i-1}),\\
\Delta s_i&=&s_i-s_{i-1}.
\end{eqnarray*}
\end{enumerate}
\end{definition}\noindent
This $P(\Delta N_i;\Delta s_i/\lambda)$ is considered as a function that gives a probability observing an increment $\Delta N_i$ when a time distance $\Delta s_i$ is given.
For any monotonically increasing series, $0=s_1<\cdots<s_n$,  the probability to find a sequence of a frequency distribution $P(\{\Delta N_i\},1\leq i\leq n;\Delta s_i/\lambda)$ can be obtained as $\prod_{i=1}^nP(\Delta N_i;\Delta s_i/\lambda)$ because a memoryless property has an exponential measure. 
A probability measure will produce a series of random variables by its independent trials. 
Those random variables are categorized to the {\it independent and identically distributed random variables}~($i.i.d.$).
It is known well that a stochastic process induced by an exponential measure, $\mu_e(s,\lambda)$, a frequency distribution will form a Poisson process.
Therefore, the following lemma can be stated:
\begin{lemma}{\bf stochastic process induced by exponential measure}\label{pp2}\\
For a stochastic process $A_i$, following two statements are equivalent:
\begin{enumerate}
\item A stochastic process $A_i$ is induced by an exponential measure $\mu(s;\lambda)$.
\item A frequency distribution $N(s)$ induced by the random variables $A_i$ forms a Poisson process.
\end{enumerate}
\end{lemma}
\begin{proof}
$1\Rightarrow 2$:
At first, random variables $\{X_i\}$ for $i=1,2,\cdots$, whose probability measure is $\mu(s;\lambda)$ are introduced. Then a parameter $0<\Delta s=s_i-s_{i-1}$ is divided into $m$ peaces, then each length becomes $ds=\Delta s/m$ and an expected number of random variables $X_i$ in $ds$, the frequency distribution, can be expressed as $\lambda/ds=\lambda m/\Delta s$. The frequency distribution is independent each other and can be obtained as $\mu(0;\lambda m/\Delta s)=\Delta s/(\lambda m)$, because of a memoryless property. Then the probability to observe the frequency distribution $\Delta N$ can be expressed after taking a limitation $m\rightarrow\infty$ as
\begin{eqnarray*}
P(\Delta N;\Delta s/\lambda)&=&\lim_{m \to \infty}~_mC_{\delta_N}
\left(\frac{\Delta s/\lambda}{m}\right)^{\Delta N}\left(1-\frac{\Delta s/\lambda}{m}\right)^{m-\delta N},\\
&=&\frac{(\Delta s/\lambda)^{\Delta N}}{{\delta N}!}\exp{\left(-\frac{\Delta s}{\lambda}\right)},
\end{eqnarray*}
where $~_aC_b=a!/k!(a-b)!$ is the combinatorial coefficient to take $b$ out of $a$ for ~ $0\leq b\leq a\in\Z$.
Therefore, the random variables $\{N(s)\}$ satisfies the requirement ${\it 3}$ in the {\bf Definition~\ref{pp}}.
Requirements ${\it 1}$ and ${\it 2}$ are trivially satisfied. Therefore, $1\Rightarrow 2$ is proved.
\\
$2\Rightarrow 1$: In the stochastic process $A_i$ induced by the probability measure $p(s)ds$, the measure $p(s)ds$ gives  the probability not to find any random variable $X_i$ from $0$ to $s$, then find one of some $X_i$ during next $ds$, which can be expressed as 
\begin{eqnarray*}
p(s)ds&=&
P(0;{s/\lambda})P(1;{ds/\lambda})\\
&=&
\exp{\left(-\frac{s}{\lambda}\right)}\times
\frac{ds}{\lambda}\exp{\left(-\frac{ds}{\lambda}\right)},\\
&=&\frac{1}{\lambda}\exp{\left(-\frac{s}{\lambda}\right)}ds+{\cal O}(ds^2),\\
&\simeq&\mu(s;\lambda).
\end{eqnarray*}
Therefore, the probability measure inducing $A_i$ is nothing but the exponential measure.
\end{proof}

%%%%%%%%%%%%%%%%%%%%%%%%%%%%%%%%
% Appendix C
%%%%%%%%%%%%%%%%%%%%%%%%%%%%%%%%
\section{Proof of Remark \ref{brownEM}}\label{ap3}
\noindent
{\bf Remark \ref{brownEM}}{\bf (Brownian motion induced by exponential measure)}(repeat)\\
Suppose that the Poisson process is induced by the exponential measure $\mu_e(s,\lambda)$ with a series of random variables, $\{X_s\}_{s=1,2,\cdots}$, and 
\begin{eqnarray*}
Z_n&=&\frac{1}{\sqrt{n}}\sum_{s=1}^n(X_s-\lambda).
\end{eqnarray*}
The random variable $Z_n$ is know to show a convergence in law to a normal distribution with a mean value at zero and a variance $\lambda$. 
In particular, for $-\infty<a<b<\infty$,
\begin{eqnarray}
\displaystyle \lim_{n \to \infty}P(a\leq Z_n\leq b)&=&\frac{1}{\sqrt{2\pi\lambda}}\int_a^b \exp{\left(-\frac{x^2}{2\lambda}\right)}dx,~~\label{brown2}
\end{eqnarray}
can be obtained.  
Brownian motion is induced by this random variable.
\begin{proof}
For the stochastic process induced by the exponential measure, each $X(s)$ is $i.i.d.$ and the frequency distribution forms a Poisson distribution as stated in {\bf Lemme \ref{pp2}}. 
For the random variable $Z_n$, mean value  and variance are given as $E[X_n]\rightarrow 0~(a.s.)$, and $V[X_n]\rightarrow \lambda<\infty~(a.s.)$, respectively, due to the law of large number. 
Therefore, for the limit $n\rightarrow\infty$, the equation (\ref{brown2}) follows from the central limit theorem.
Therefore, the Gaussian probability measure can be defined as
\begin{eqnarray*}
P(B_\lambda\in dx)&=& \lim_{n \to \infty}P(x\leq Z_n\leq x+dx)
=\frac{1}{\sqrt{2\pi\lambda}}\exp{\left(-\frac{x^2}{2\lambda}\right)}dx.
\end{eqnarray*}
Therefore, the Poisson process converges in low to the Brownian motion.
\end{proof}

%%%%%%%%%%%%%%%%%%%%%%%%%%%%%%%%
% Appendix D
%%%%%%%%%%%%%%%%%%%%%%%%%%%%%%%%
\section{Stochastic Integration}\label{ap4}
This appendix is devoted to introduce and summarize the stochastic integrations and stochastic differential equations without proofs.
\subsection{It{\^ o} formula}\label{A1}
\begin{definition}\label{brown}{\bf Brownian motion}\\
If a set of Random variable, ${\bm B}(t)=\{B(t),t\in\R\}$ satisfies the following properties:
\begin{enumerate}
\item ${\bm B}(t)$ is a Gaussian system with $E[B(t)]=0$ for every $t$,
\item $B(0)=0$,
\item $E[(B(t)-B(s))^2]=|t-s|$
\end{enumerate}
it is called Brownian motion\cite{hidasisi2008}.
\end{definition}\noindent
The covariant matrix of the Brownian motion can be obtained as
\begin{eqnarray*}
\Gamma[s,t]&=&E[B(t)B(s)]=\frac{1}{2}\left(
|t|+|s|-|t-s|
\right),
\end{eqnarray*}
due to the definition.
\begin{definition}\label{ip}{\bf It{\^ o} process and It{\^ o} formula}\\
On following setups:
\begin{itemize}
\item Functions:
\begin{enumerate}
\item $x=\{x^i_t\}_{i=1,\cdots,d}\in C^1(\R^d)$ for $t\geq 0$, 
\item $\forall\varphi=\varphi(x)\in C^2(\R)$.
\end{enumerate}
\item Two stochastic processes:
\begin{enumerate}
\item $N$-independent Brownian motion: \\
${\bm B}(t)=\{ B^k(t)\}_{k=1,\cdots,N}$.
\item stochastic processes: \\
$\{A^i(t)\}_{i=1,\cdots,d}$, where
$\forall i, \omega,~~A_\omega^i(0)=0 \wedge A_\omega^i(t)<\infty$,
\end{enumerate}
where $\omega$ is a index for sample paths. The index $\omega$ will be  often omitted.
\end{itemize}
It{\^ o} process is defined as random variables of $\{X_t^i \}_{i=1,\cdots,d}$, $t\geq 0$, which follows:
\begin{eqnarray*}
X^i_t&=&X^i_0+\int_0^t {\bm\sigma}^i\cdot d{\bm B}(s)+A^i(t),
\end{eqnarray*}
where ${\bm\sigma}^i_k(s)=\{\sigma^i_k(s)\in\LL^2\}_{k=1,\cdots,N}$,~
$(\LL^2$ is a set of square integrable functions.$)$
Then the It{\^ o} formula\cite{ito1944} is introduced as
\begin{eqnarray*}
\varphi(X_t)&=&\varphi(X_0)+\int^t_0\frac{\partial\varphi}{\partial x^i}
(X_s){\bm\sigma}^i(s)\cdot d{\bm B}(s)
+\int^t_0
\frac{\partial\varphi}{\partial x^i}(X_s)dA^i_s
+\frac{1}{2}\frac{\partial^2\varphi}{\partial x^i\partial x^j}(X_s)
{\bm \sigma}^i(s)\cdot{\bm \sigma}^j(s)ds.
\end{eqnarray*}
\end{definition}\noindent
Here, integration with respect to $dA^i_s$ must be understood as the Riemann-Stieltjes integration, and the stochastic integration with respect to the Brownian motion ${\bm\sigma}^i\cdot d{\bm B}(s)$ must be understand as the It\^{o} integral, that is defined as;
\begin{eqnarray*}
\int_0^t {\bm\sigma}^i(s)\cdot d{\bm B}(s)&=&
\lim_{n\rightarrow\infty}\sum_{k=0}^{n-1}
{\bm\sigma}^i(s_k)\left[
{\bm B}_{k+1}-{\bm B}_k
\right],
\end{eqnarray*}
where $\{s_i\}_{i=1,\cdots,n}$ are the dissection $0=s_1<\cdots<s_n=t$, and ${\bm B}_k={\bm B}(s_k)$. Each ${\bm B}_k$ will be independent Gaussian noise. 
The It\^{o} integral is the martingale because ${\bm\sigma}^i(s_k)$ is only depends on the process $[0,s_k]$.
Another definition of the stochastic integration can be given as;
\begin{eqnarray*}
\int_0^t {\bm\sigma}^i(s)\circ d{\bm B}(s)&=&
\lim_{n\rightarrow\infty}\sum_{k=0}^{n-1}
{\bm\sigma}^i(s_{k+1/2})\left[
{\bm B}_{k+1}-{\bm B}_k
\right],
\end{eqnarray*}
where $s_{k+1/2}=(s_{k+1}+s_{k})/2$. It is known as the Stratonovich integral, which is not martingale, but sub-martingale. We use the It\^{o} convention as the stochastic integral in this report, unless otherwise stated.
The It{\^ o} formula can be expressed as;
\begin{eqnarray*}
&~&
\varphi(X(t))=\varphi(X_0)+\int^t_0\frac{\partial\varphi}{\partial x^i}
(X(s)){\bm\sigma}^i(s)\cdot d{\bm B}(s)
%\\&~&
+\left[\int^t_0
\frac{\partial\varphi}{\partial x^i}(X(s))\beta^i(s)
+\frac{1}{2}\frac{\partial^2\varphi}{\partial x^i\partial x^j}(X(s))
{\bm \sigma}^i(s)\cdot{\bm \sigma}^j(s)\right]ds,
\end{eqnarray*}
when the stochastic process $A^i(s)$ is written as;
\begin{eqnarray*}
A^i(t)-A^i(0)&=&\int_0^t \beta^i(s)ds.
\end{eqnarray*}
This formula can be understood from a Taylor expansion of $\varphi$ with simple rules ({\it It\^{o} rules}) as 
\begin{eqnarray*}
\left\{
  \begin{array}{cll}
    d{\bm B}(s)\cdot d{\bm B}(s)&=&{\bm I}ds,\\
    d{\bm B}(s)(\beta^i(s) ds)&=&0,\\
    (\beta^i(s) ds)(\beta^j(s) ds)&=&0,
  \end{array}
\right.
\end{eqnarray*}
where ${\bm I}$ is a $N\times N$ unit matrix. 
In other words, they are evaluated as $d{\bm B}(s)\sim\sqrt{ds}$ and $dA(s)\sim ds$ in the expansion and taken up to order $ds$ such as;
\begin{eqnarray*}
d\varphi(X(t))&=&\frac{\partial\varphi}{\partial x^i}
(X(t))dX^i(t)
+\frac{1}{2}\frac{\partial^2\varphi}{\partial x^i\partial x^j}(X(t))
dX^i(t)dX^j(t)+\cdots,
\end{eqnarray*}
where
\begin{eqnarray*}
dX^i(t)&=&dA^i(t)+{\bm\sigma}^i(t)\cdot d{\bm B}(t)
=\beta^i(t)dt+{\bm\sigma}^i(t)\cdot d{\bm B}(t),\\
dX^i(t)dX^j(t)&=&
({\bm\sigma}^i(t)\cdot d{\bm B}(t))({\bm\sigma}^j(t)\cdot d{\bm B}(t))
=({\bm\sigma}^i(t)\cdot{\bm\sigma}^j(t))dt.
\end{eqnarray*}
Then, the It{\^ o} formula follows from them. 
Above expressions must be understood as a shorthand for the integration form, and a following remark for one and two random variables $X^i(t)$ can also be written as;
\begin{eqnarray*}
d\left(f(t)X^i(t)\right)&=&f(t)(dX^i(t))+(df(t)/dt)X^i(t) dt,\\
d\left(X^i(t) X^j(t) \right)&=&
X^i(t)(dX^j(t))+(dX^i(t))X^j(t)+(dX^i(t))(dX^j(t)),
\end{eqnarray*}
where $f(t)\in C^2$.
Above formulae are referred as to  the {\it Stochastic Leibniz rule}.
Integration forms of above expressions can be given as:
\begin{eqnarray*}
f(t)X^i(t)-f(0)X^i(0)
&=&\int_0^t f(s)dX^i(s)+\int_0^t\frac{df(s)}{ds}X^i(s) ds,\\
&=&\int_0^t\left(\beta^i(s)+\frac{df(s)}{ds}X^i(s)\right)ds
+\int_0^t{\bm \sigma}^i(s)\cdot d{\bm B}(s),\\
X^i(t) X^j(t)-X^i(0) X^j(0)&=&
\int_0^t\left(X^i(s)\beta^j(s)+X^j(s)\beta^i(s)\right)dt
+\int_0^t\left(X^i(s){\bm \sigma}^j(s)
+X^j(s){\bm \sigma}^i(s)\right)
\cdot d{\bm B}(s)
\\&~&\hskip 8cm
+\int_0^t ({\bm\sigma(t)}^i\cdot{\bm\sigma(t)}^j)dt.
\end{eqnarray*}
The second line is just a special case of the It{\^o} formula with $\varphi(X(t))=X^i(t) X^j(t)$.

The stochastic integration of a function $f(x)$ by the Brownian motion is given as;
\begin{eqnarray}
\int_0^tf(s)dB(s)&=&B(t;\sigma^2),\label{stin1}
\end{eqnarray}
where $B(t;\sigma^2)$ is the Gaussian process with a variance of 
\begin{eqnarray}
\sigma^2=\int_0^tf(s)^2ds,\label{stin2}
\end{eqnarray}
and mean value of zero. 
One of an important property of the Gaussian distribution is a reproducing property, which means the convolution of two Gaussian distributions make a Gaussian distribution.

\subsection{Iterated Stochastic Integration}\label{A2}
First, let us introduce an iterated classical (non-stochastic) integration of one-forms $\omega_{i=1,\cdots,n}$ along the curvilinear path $\{\gamma\in\Gamma\}$ as defined in Section \ref{sec3-1}.
A pull of $\omega_i$ using the path $\gamma$ is expressed as $\gamma^*\omega_i=f_i(t)dt$, and then, the iterative integration can be introduced as;
\begin{eqnarray*}
\int_\gamma\omega_1\cdots\omega_n
&=&\int_{\delta_n}\gamma^*\omega_1\cdots\gamma^*\omega_n,
=\int_0^Tdt_n
\left(f_n(t_n)
\int_0^{t_n}dt_{n-1}
\left(%f_{n-1}(t_{n-1})
\cdots \left(f_2(t_2)
\int_0^{t_2}dt_1f_1(t_1)
\right)\cdots\right)\right),
\end{eqnarray*}
using simplexes $\delta_n$ defined as $0\leq t_1\leq\cdots\leq t_n\leq T$.
A natural extension of this iterative integration to the iterative stochastic integrations might be;
\begin{eqnarray*}
I_n\left(\{f_i\}\right)
&=&\int_0^TdB(t_n)
\Bigl(f_n(t_n)\int_0^{t_n}dB(t_{n-1})
\Bigl(
\cdots\Bigl(f_2(t_2)
\int_0^{t_2}dB(t_1)f_1(t_1)
\Bigr)\cdots\Bigr)\Bigr).
\end{eqnarray*}
The existence of this multiple integrations can be shown using approximation of $f_i$ being step functions\cite{hida1980brownian}. 
It is known that the stochastic integration for the unit function can be represented by using the parameterized-Hermite polynomials as follows\cite{hida1980brownian};
\begin{eqnarray*}
I_n(\{{\bm 1}\})
&=&\int_0^TdB(t_n)
\left(\int_0^{t_n}dB(t_{n-1})\left(\cdots
\left(\int_0^{t_2}dB(t_1)\right)
\cdots\right)\right)
=H_n\left(B(T)-B(0),T-0\right),\label{ISI2}
\end{eqnarray*}

where $H_n$ can be expressed as;
\begin{eqnarray*}
H_n(x,t)&=&\frac{(-t)^n}{n!}\exp{\left(
\frac{x^2}{2t}
\right)}
\frac{d^n}{dx^n}
\exp{\left(-
\frac{x^2}{2t}
\right)}.\label{HEP1}
\end{eqnarray*}

One of interesting examples of the iterated stochastic integration is a {\it stochastic area} introduced by L\`{e}vy\cite{levy1940}:
\begin{definition}{\bf Stochastic Area}\label{STA}
\begin{eqnarray*}
S_\omega(T)
&=&\frac{1}{2}\int_0^T
\left[dB^1_\omega(t),dB^2_\omega(t)\right]
=\frac{1}{2}\left\{
\int_0^T\left(\int_0^tdB^1_\omega(s)\right)dB^2_\omega(t)-
\int_0^T\left(\int_0^tdB^2_\omega(s)\right)dB^1_\omega(t)
\right\},\\
&=&\frac{1}{2}\int_0^T\left(
B^1_\omega(t)dB^2_\omega(t)-
B^2_\omega(t)dB^1_\omega(t)
\right),
\end{eqnarray*}
where ${\bm B}_\omega(t)=(dB^1_\omega(t),dB^2_\omega(t))$, with $t\geq0$, is a $2$-dimensional Brownian motion. 
\end{definition}
This quantity can be interpreted as the area surrounding by trajectory of ${\bm B}_\omega$ and straight line connecting between initial and final points of the trajectory. 
%%%%

%%%%%%%%%%%%%%%%%%%%%%%%%%%%%%%%
% Appendix E
%%%%%%%%%%%%%%%%%%%%%%%%%%%%%%%%
\section{Stochastic Differential equation}\label{ap5}
With the same setup as in Appendix \ref{ap4}, one can introduce the stochastic differential equations (SDE) as 
\begin{eqnarray}
dX(t)^i={\bm \sigma}^i(t,X(t))d{\bm B}(t)+\beta(t)^i(t,X(t))dt.\label{SDE}
\end{eqnarray}
Here, ${\bm \sigma}^i$ and $\beta(t)^i$ are referred to as  {\it diffusion} and  {\it drift coefficients}, respectively. 
Exact meaning of equation (\ref{SDE}) might be clear after cast it into a integration form. 
A definition of a solution of the SDE can be given as follows: 
\begin{definition}{\bf Solution of the SDE}\label{SSDE}\\
Let us consider a stochastic process $X$ defined on a probabilistic space $(\Omega,\A, P)$, which is assumed to be $(\A)$-adaptive and measurable $\R^d$-valued continuous stochastic process. 
The solution of the SDE (\ref{SDE}) is the random variables, $X=\{X(t)\}_{t\geq0}$, which satisfied following two conditions:
\begin{enumerate}
\item For any $i$ and $k$, ${\bm\sigma}^i_k(t,X(t))\in\LL^2(\A(t))$, and 
\begin{eqnarray*}
\int_0^T|\beta^i_a(t,X(t))|dt<\infty, 
\end{eqnarray*}
for any $T\geq0$ (a.s.).
\item $X(t)$ satisfies a following relation including a stochastic integral:
\begin{eqnarray*}
X(t)&=&x+\int_0^t{\bm \sigma}^i(t,X(t))d{\bm B}(t)
+\int_0^t\beta(t)^i(t,X(t))dt.
\end{eqnarray*}
\end{enumerate}
When the existence of the Brownian motion is assumed, it is referred to the {\it strong solution}.
\end{definition}\noindent
Here, we will give two examples and their (strong) solutions.
\begin{example}{\bf Linear SDE}\label{LSDE}\\
The solution of the linear is SDE defined as;
\begin{eqnarray*}
dX(t)&=&\left(
\tilde{\alpha}(t)+\tilde{\beta}(t)X(t)
\right)dt
+
\left(
\tilde{\gamma}(t)+\tilde{\delta}(t)X(t)
\right)dB(t),
\end{eqnarray*}
where $\{\tilde{\alpha}(t),\tilde{\beta}(t),\tilde{\gamma}(t),\tilde{\delta}(t)\}$ are given integrable functions.
The solution of this SDE can be written as;
\begin{eqnarray*}
X(t)&=&U(t)\Big\{
X_0+\int_0^t\left(
\tilde{\alpha}(s)
-\tilde{\delta}(s)\tilde{\gamma}(s)\right)U(s)^{-1}ds
%\nonumber\\&~&
+\int_0^t\tilde{\gamma}(s)U^{-1}(s)dB(s)
\Big\},\\
U(t)&=&\exp{\left[
\int_0^t\left(\tilde{\beta}(s)
-\frac{1}{2}\tilde{\delta}^2(s)\right)ds
+\int_0^t\tilde{\delta}(s)dB(s)
\right]}.
\end{eqnarray*}
\end{example}\noindent
\begin{example}{\bf Pinned Brownian motion}\label{pbm}\\
One dimensional SDE of
\begin{eqnarray*}
dX(t)&=&dB(t)+\frac{a-X(t)}{1-t}dt,
\end{eqnarray*}
$0\le t<1$, and $X(0)=0$ is called  a {\it pinned Brownian motion}.
\end{example}\noindent
It is easily confirmed that the solution of this {\it SDE} is
\begin{eqnarray*}
X(t)&=&at+(1-t)\int_0^t\frac{dB(s)}{1-s}=at+B(t)-tB(1)
\end{eqnarray*}
The stochastic process $X(t)$ is starting from the origin at $t=0$ and arrived at $a$ at $t=1$ with the mean convergence.
It is easily understood that time-reversed process, $X(1-t)$, is the same process as original one, because these two processes have the same mean value and variance.

%%%%%%%%%%%%%%%%%%%%%%%%%%%%%%%%
% Appendix F
%%%%%%%%%%%%%%%%%%%%%%%%%%%%%%%%
\section{Functional of Brownian motion and white noise}\label{ap6}
Basic definitions and several remarks are given in this section without proofs.
Detailed explanations can be found in \cite{hida1980brownian}. 
\subsection{one-dimensional case}\label{A4-1}
Let us consider distributions of stochastic processes in the sense of Schwartz. 
The real Hilbert-space of square integrable functions is taken as a set of test functions as $E\subset H=L^2(T\in\R)$. 
A space of distribution, $E^*$, is a dual space of $E$, which has a relation of
\begin{eqnarray*}
E\subset H\subset E^*.
\end{eqnarray*}
A standard bilinear form is written as
$
\left\langle x,\xi\right\rangle,
$
where $\xi\in E^*,~x\in E$. 
When $\xi\in H$, this is nothing but an inner product of two vectors on the Hilbert space. 
A norm of two elements in the Hilbert space is denoted as
$\parallel x\parallel^2=\left\langle x,x\right\rangle$. 
A set of stochastic processes defined on some probabilistic space is considered to be parameterized by $x\in E$ such as $\{X(x)|x\in E\}$. 
For the probabilistic distribution (measure) $\mu$ on $E^*$, the characteristic functional is defined as
\begin{eqnarray*}
C_X(x)&=&\int_{E^*}
\exp{\left(i\left\langle \xi,x\right\rangle\right)}
\mu(d\xi).
\end{eqnarray*}
The characteristic functional must have following characteristics:
\begin{enumerate}
\item $C_X$ is continuous on $E$,
\item $C_X$ is positive definite,
\item $C_X(0)=1$.
\end{enumerate}
The measure $\mu(d\xi)$ on $E^*$ is uniquely determined by the characteristic functional, thanks to the Bochner--Minlon theorem.

A vector space of the standard bilinear form, $\left\langle x,f\right\rangle,~(f\in L^2(\R^1))$ is considered.
This vector space becomes the Hilbert space with the inner product as the covariance integral and the subspace of $L^2(\mu)$.
We denote this Hilbert space as $H_1^{(1)}$ and its dual space as $H_1^{(-1)}$.
A subscript $1$ means the Hilbert space defined on a one-parameter space.  $(\pm1)$ is put for future convenience.
Isomorphism $H_1^{(1)}\cong L^2(\R^1)$ can be obtained under the isomorphic map $f\leftrightarrow\left\langle x,f\right\rangle$.
Thus the Gel'fand triple of $H_1^{(1)}\subset H\subset H_1^{(-1)}$ is obtained.

\begin{example}{\bf White noise}\label{WH}\\
The measured space induced by the characteristic function of 
\begin{eqnarray*}
C(x;\sigma^2)&=&
\exp{\left(-\frac{\sigma^2}{2}\parallel x\parallel^2\right)}
\end{eqnarray*}
is named {\it White noise} with a variance of $\sigma^2$, which denoted as $\dot{B}(t)$.
Here $\|\bullet\|$ is a Hilbertian norm.
\end{example}\noindent
The probabilistic distribution of white nose is the Gaussian distribution with mean value of zero and variance $\sigma^2$.
It is easily confirmed that the white noises with different variances are mutually disjoint.
\begin{example}{\bf Second derivative of Brownian motion }\label{WH2}\\
The characteristic functional of
\begin{eqnarray*}
C'(\xi)&=&
\exp{\left(-\frac{1}{2}\|\xi'\|^2\right)}
\end{eqnarray*}
induces a process of second derivative of the Brownian motion, 
$\ddot{B}(t)$ for $-\infty<t<\infty$.
\end{example}\noindent
The measure of this process, $\mu'$, is defined on the same space as that of the original white noise, however those are the mutually disjoint each other. If the support of the original measure is $X$, the measure $\mu'$ gives
\begin{eqnarray*}
\mu(X)=1\Rightarrow\mu'(X)=0,
\end{eqnarray*}
and vice versa.

\subsection{multi-dimensional case}\label{A4-2}
The definition of Brownian motion and white noise can be extended to a multi-dimensional parameter space. 
\begin{definition}\label{mbrown}{\bf Brownian motion on the Multi-dimensional parameter space}\\
If a set of Random variables, ${B}(x)=\{B(x),x\in\R^d\}$ satisfies the following properties:
\begin{enumerate}
\item ${B}(x)$ is a Gaussian system,
\item $E[B(x)]=0$,
\item Covariant matrix is
\begin{eqnarray*}
\Gamma[x,y]&=&E[B(x)B(y)]
=\frac{1}{2}\left(
|x|+|y|-|x-y|
\right),
\end{eqnarray*}
\end{enumerate}
it is called Brownian motion.
\end{definition}\noindent
From the definition, $B({\bf o})=0$ can be obtained, where ${\bf o}=(0,0,\cdots,0)$ is the origin of the parameter space. 
A proof of existence of multi-dimensional Brownian motions can be found in \cite{hidasisi2008}.

In order to define the multi-dimensional white noise, Hilbert space, its dual space and Gel'fand triple defined in the previous subsection are extended to
$E_d\subset H_d\cong L^2(x\in\R^d)$, $E_d^*$ and $E_d\subset H_d\subset E_d^*$.
On the dual space $E^*_d$, a completely additive measure space $(E^*_d, \B, \mu)$ can be constructed. 
According to notations for th one-parameter case, we can denote the Gel'fand triple as
\begin{eqnarray*}
H_d^{(1)}\subset H_d\subset H_d^{(-1)}
\end{eqnarray*}
Existence of the completely additive measure associated with the characteristic function of
\begin{eqnarray*}
C(\xi)&=&\int_{E^*_d} \exp{\left(i\langle x,\xi \rangle \right)}d\mu(x),\nonumber\\
&=&\exp{\left( -\frac{1}{2}\|\xi\|^2\right)},\label{charact2}
\end{eqnarray*}
where $\xi\in E_d$ and $\|\bullet\|$ is a Hilbertian norm, is ensured by the Bochner--Minlos theorem.
\begin{definition}\label{mwhitenoise}{\bf Multi-dimensional white noise}\\
A stochastic process induced by this characteristic function of eq.{\rm(\ref{charact2})} is called {\it multi-dimensional white noise}.
The white noise induced by the multi-dimensional Brownian motion ${B}(x)$ is denoted as $B_{;\mu}(x)$.
\end{definition}

\subsubsection{Polynomial of white noise}\label{WNpoly}
Further extension of the white-noise space for  polynomials of the white noise is possible\cite{hida1980brownian,hidasisi2008}.
At first, a $n$-dimensional white-noise functional space is introduced as
\begin{enumerate}
\item Test functional space $H_d^{(n)}$: a symmetric $(n+1)/2$ dimensional Sobolev space on $\R^n$.
\item Hilbert space $H_n$: a symmetric tensor product of $n$ square-integrable functions.
\item White-noise functional space $H_d^{(-n)}$: a dual space of the test functional space, which is an isomorphic with a symmetric $-(n+1)/2$ dimensional Sobolev space on $\R^n$.
\end{enumerate}
The Gel\rq{}fand triple can be obtained  as
$
H_d^{(n)}\subset H_d\subset H_d^{(-n)}.
$
Power functionals of the white noise are a element of the dual space ($B^n_{;\bullet}\in H_d^{(-n)}$). 

A polynomial of the white-noise functional can be defined on the space $\bigotimes_n  c_n H_n^{(-n)}$ with the Gel\rq{}fand triple of
\begin{eqnarray*}
\bigotimes_n  \frac{1}{c_n} H_n^{(n)}\subset (L^2)\subset\bigotimes_n  c_n H_n^{(-n)},
\end{eqnarray*}
where $c_n$ is a positive monotonically non-increasing series. 

%%%%%%%%%%%%%%%%%%%%%%%%%%%%%%%%
% Appendix G
%%%%%%%%%%%%%%%%%%%%%%%%%%%%%%%%
\section{{\bf Path integral with Gaussian measure}}\label{ap7}
Let us start from the Lagrangian (\ref{qpaLag}) and defined a path integral with the Gaussian measure under the time-slicing approximation\cite{KUMANOGO2004197} as follows:
the setup defined in section \ref{sec4} is used in the appendix.
\begin{definition}{\bf (Gaussian path-integral)}\\
From the Lagrangian (\ref{qpaLag}), a transition amplitude from point  $\gamma(0)$ to $\gamma(T)$ is defined as
\begin{eqnarray*}
\psi(\gamma(T))&=&\int \exp{\left(
\frac{i}{\hbar}\int\I(\gamma)
\right)} \mu_G(d\gamma),\\
&=&C(\epsilon)\prod_k\int_{-\infty}^\infty C(\epsilon)
\exp{\left(
\frac{i}{\hbar}\frac{m}{2}\frac{\left(\gamma_{k+1}-\gamma_{k}\right)^2}{\epsilon}-\frac{i\epsilon}{\hbar}
V\left(\frac{\gamma_{k+1}-\gamma_k}{2}\right)
-\frac{1}{2\sigma^2}(\gamma_{k+1}-\gamma_{k})^2
\right)}
d\gamma_k.
\end{eqnarray*}
\end{definition}
\noindent
This path-integral is defined using the time-slicing approximation, which divides a time period from $0$ to $T$ with a width $\epsilon$.
Here $C(\epsilon)$ is an appropriate normalization constant as a function of $\epsilon$ and $\gamma_i=\gamma(\tau_i)$.
A relation between the flat ``measure" and the Gaussian measure is given as
\begin{eqnarray*}
\mu_G(d\gamma_k)&=&\exp{\left[
-\frac{1}{2\sigma^2}\left(\gamma_{k+1}-\gamma_{k}\right)^2
\right]}d\gamma_k,
\end{eqnarray*}
where a  $\sigma^2$ is a variance of the Gaussian distribution with a dimension of length, which may be fixed by characteristics the SMM-space.
The ``measure" is flat when setting $\sigma^2\rightarrow+\infty$.
If the time-slicing width $\epsilon$ is a real number, the integration with the flat \lq\lq{}meadure\rq\rq{} does not converge.
In a standard method to evaluate the path integral in the quantum field theory  (see, for instance, Chapter 9 in a textbook\cite{opac-b1131978}), the integration is treated formally as $\epsilon$ has a negatively finite imaginary-part.
In contrast with the standard method, the Gaussian path-integration converges with a real values of a time-slicing in the SMM-space.
\begin{remark}{\bf (Effect of Gaussian measure)}\\
The Gaussian measure gives the same transition amplitude as that with the flat ``measure''.
\end{remark}\noindent
\begin{proof}
Let us pick up any one of the integrands, say a $k$'th time-slice, in infinite product of integrations, and expand the potential and transition amplitude with respect to the small parameter $\gamma_k-\gamma_{k-1}$ according to ref.\cite{opac-b1131978}, one can get
\begin{eqnarray*}
\psi\left(\tau_{k+1}:\tau_k\right)
&=&C(\epsilon)\int_{-\infty}^\infty d\gamma_{k}
\exp{\left[
\frac{m}{2\hbar}\left(\gamma_{k}-\gamma_{k-1}\right)^2\left(\frac{i}{\epsilon}-\frac{\hbar}{m\sigma^2}\right)\right]}
\left[1-\frac{i\epsilon}{\hbar}V(\gamma_{k-1})+\cdots\right]\\&~&~~~~~~~~~~~\times
\left[
1+(\gamma_{k}-\gamma_{k-1})\frac{\partial~}{\partial\gamma_{k}}
+\frac{1}{2}(\gamma_{k}-\gamma_{k-1})^2\frac{\partial^2~}{\partial\gamma_{k}^2}+\cdots\right]
\psi\left(\tau_{k}:\tau_{k-1}\right),\\
&=&\left[
1+\frac{i\epsilon\hbar}{2m}\frac{\partial^2~}{\partial\gamma_{k}^2}
-\frac{i\epsilon}{\hbar}
V(\gamma_{k-1})\right]\psi\left(\tau_{k}:\tau_{k-1}\right)+\mathcal{O}(\epsilon^2),
\end{eqnarray*}
where the normalization constant is set to be 
\begin{eqnarray*}
C(\epsilon)&=&
\sqrt{
\frac{m}{2\pi i\epsilon\hbar}
}.
\end{eqnarray*}
The integration is performed with a assumption that $\epsilon$ is a real positive-number.
The transition amplitude is independent of a value of $\sigma$.
\end{proof}

%%%%%%%%%%%%%%%%%%%%%%%%%%%%%%%%
% Appendix H
%%%%%%%%%%%%%%%%%%%%%%%%%%%%%%%%
\section{Rotational invariance of white noise}\label{ap8}
A rotational group can be introduced on the multi-dimensional white-noise space defined above.
A following definition is based on Yoshizawa\cite{yoshizawa1970} and Hida\cite{hidasisi2008}.
\begin{definition}\label{rotaion}{\bf Rotational group on ${\bm E_d}$}\\
A rotational group, $g$, is defined as a continuous linear homeomorphism acting on $E_d$, which satisfies
\begin{eqnarray*}
\parallel gx \parallel&=&\parallel x \parallel,
\end{eqnarray*}
where $x\in E_d$.
\end{definition}\noindent
It is obvious that a collection of $g$ forms a group. 
Even though this group is named a {\it rotational group} by Yoshizawa, it includes wider class of orthogonal groups in reality.
Therefore, the group is denoted as $O(E_d)$ hereafter.  
This group induces an adjoint transformation $g^*$ on the dual space $E^*_d$ such as
\begin{eqnarray*}
\langle \xi, gx \rangle&=&\langle g^*\xi, x \rangle,
\end{eqnarray*}
where $x\in E_d$ and $\xi\in E^*_d$.
This group is denoted as $O^*(E^*_d)$.
Following statements can easily confirmed:
\begin{itemize}
\item $O^*(E^*_d)$ forms a group.
\item $g^*$ is also a continuous linear homeomorphism.
\item $O^*(E^*_d)$ is isomorphic to $O(E_d)$.
\end{itemize}
Following theorem is essential for a application of the white noise analysis for physics applications.
\begin{theorem}\label{wnrotmu}{\bf rotational invariance of the white-noise measure}\\
The white noise measure $\mu$ is invariant under the group operation $g^*\in O^*(E^*_d)$.
\end{theorem}\noindent
It can be proved by checking an invariance of the  characteristic functional after applying an operation $g^*$ on $\xi\in E^*_d$.
Note that this theorem is valid for the Lorentz transformation $\Lambda$ on the SLM-space due to the fact $\Lambda\in O(E_4)$.
Therefore, one can treat the VWN, $B^{(a)}_{;\mu}(x)$ defined in section \ref{sec4-2}, as the local Lorentz-vector, since the white-nose measure is invariant under $\Lambda^*\in O^*(E^*_4)$ .
%
% BibTeX users please use
%
\bibliographystyle{elsarticle-num}
\bibliography{ref}
\end{document}